\documentclass[english,11pt,a4paper,final,reqno]{smfart}

\usepackage{smfthm}

\usepackage{lmodern}
\usepackage[utf8]{inputenc}
\usepackage[T1]{fontenc}
\usepackage[english]{babel}
\usepackage[babel]{csquotes}
\usepackage{graphicx}
\usepackage{xspace}
\usepackage{multirow}
\usepackage{multirow}
\usepackage{tikz}
\usetikzlibrary{arrows,calc,shapes}
\usetikzlibrary{patterns}
\usetikzlibrary{matrix}
\usepackage{enumerate}

\usepackage[a4paper, margin=3cm]{geometry}

\newcommand{\indic}[1]{\mathrm{\textbf{1}}_{#1}}

\newcommand{\setN}{\mathbb{N}}
\newcommand{\setZ}{\mathbb{Z}}
\newcommand{\setK}{\mathbb{K}}
\newcommand{\setR}{\mathbb{R}}
\newcommand{\setQ}{\mathbb{Q}}

\newcommand{\Sym}[1]{\mathfrak{S}_{#1}}

\newcommand{\ca}[1]{\mathcal{#1}}
\newcommand{\ov}[1]{\overline{#1}}

\newcommand{\ti}[1]{\widetilde{#1}}
\newcommand{\gr}[1]{\mathbf{#1}}

\DeclareMathOperator{\Orb}{Orb}

\DeclareMathOperator{\id}{id}

\newcommand{\timesdots}{\times\ldots\times}

\usepackage{hyperref}

\usetikzlibrary{patterns}
\usetikzlibrary{matrix}
\usepackage{mathtools}

\usepackage{esint}

\usepackage{algorithm}
\usepackage{algpseudocode}
\algrenewcommand\algorithmicrequire{\textbf{Input}}
\algrenewcommand\algorithmicensure{\textbf{Output}}

\newcommand{\mrb}{mixed radix basis\xspace}
\newcommand{\mrbs}{mixed radix bases\xspace}

\newcommand{\Mrbs}{Mixed radix bases\xspace}

\newcommand{\South}{\mathsf{S}}
\newcommand{\North}{\mathsf{N}}
\newcommand{\West}{\mathsf{W}}
\newcommand{\East}{\mathsf{E}}

\newcommand{\MarkovWeight}[1]{\mathsf{#1}}

\tikzstyle{guillpart} = [scale=0.5,baseline={(current bounding box.center)}]
\tikzstyle{guillfill}=[gray,opacity=0.5]
\tikzstyle{guillsep} = [ultra thick]
\tikzstyle{centerline} = [baseline={(current bounding box.center)}]

\begin{document}
	
	\title[Mixed radix bases: Horner, Yang, Baxter and Furstenberg]{Mixed radix numeration bases: Horner's rule, Yang-Baxter equation and Furstenberg's conjecture.}
	\author{Damien \textsc{Simon}}
	\address{Sorbonne Université and Université Paris Cité, CNRS, Laboratoire de Probabilités, Statistique et Modélisation, 4, place Jussieu, F-75005 Paris, France}
	\email{damien.simon@sorbonne-universite.fr}
	\date{2025/03/27}
	

\begin{abstract}
	Mixed radix bases in numeration is a very old notion but it is rarely studied on its own or in relation with concrete problems related to number theory. Starting from the natural question of the conversion of a basis to another for integers as well as polynomials, we use mixed radix bases to introduce two-dimensional arrays with suitable filling rules. These arrays provide algorithms of conversion which use only a finite number of Euclidean division to convert from one basis to another; it is interesting to note that these algorithms are generalizations of the well-known Horner's rule of quick evaluation of polynomials. The two-dimensional arrays with local transformations are reminiscent of statistical mechanics models: we show that changes between three numeration bases are related to the set-theoretical Yang-Baxter equation and this is, up to our knowledge, the first time that such a structure is described in number theory. As an illustration, we reinterpret well-known results around Furstenberg's conjecture in terms of Yang-Baxter transformations between mixed radix bases, hence opening the way to alternative approaches.
\end{abstract}

\maketitle

\section{Introduction}

\subsection{The background}
The purpose of the present paper is to exhibit a new link between two different fields: number theory and exactly solvable models in statistical mechanics. We believe that the present construction may produce new results in number theory. Indeed, the major difficulty in number theory is that the quantity of available tools is rather small due to the very elementary structure of $\setZ$. On the algebraic side, there is the construction of ideals, which establishes a similarity with function fields: it gives rise to the finite fields $\setZ/p\setZ$ when $p$ is prime and to the construction of $p$-adic spaces such as $\setZ_p$ or $\setQ_p$. In many situations, the prime number $p$ is generic but fixed and the structure of the previous structures is studied in detail. There also exist situations in which one considers all the possible $p$ simultaneously such as Berkovich spaces or the use of Hasse principle to solve Diophantine equations. The present work falls within this latter case: instead of decomposing integers in a given basis, we will rather consider how decompositions of numbers in various bases interact with each other. When considering simultaneous decompositions of integers in various constant numeration bases, let say numeration in constant bases $p$ and $q$, the surprise described in the present paper is that the formulae of change of bases become related to a well-known tool of exactly solvable models in statistical mechanics: the Yang-Baxter equation. 

\subsection{A two-dimensional view on \mrbs}
To achieve this relation, we need to lift classical constant radix numeration bases, i.e. decomposition of numbers as $\sum_{i} a_i r^i$ with $0\leq a_i<r$, to a wider class of numeration bases called \mrbs, which are as old as numeration systems and whose most prominent actual example is the second-minute-hour-day unit system of time. The major introduction of the present paper presented in section~\ref{sec:mixedradix} is to decompose any change from base $p$ to base $q$ into elementary "local" changes of bases that relate neighbouring \mrbs. As such, it does not use any highly-technical constructions but the new part is that the collection of these elementary changes of bases exhibits a two-dimensional lattice structure with the Markov property which has a interpretation as a two-dimensional model of statistical mechanics.

A constant numeration basis in base $p$ has a natural one-dimensional interpretation as a sequence of digits indexed by $\setN$. While segments correspond to integer numbers, completions to half-lines correspond to $p$-adic numbers. Using \mrbs interpolating between numeration in bases $p$ and $q$, we obtain two-dimensional representations of integer numbers with boundaries corresponding to the two initial bases. A \mrb then corresponds to a broken line in such a rectangle. The present paper describes the construction in details, with various examples. In particular, this leads to nice algorithmic properties described in section~\ref{sec:algo}. From this point of view, the new algorithm require only a finite number of divisions and remainders and replace most equations by a (easy) filling of a domain by tiles.

The present paper focuses on the integers but the same construction also holds for polynomial ring $\setK[X]$ without surprise. We also describe the formalism of \mrbs for polynomials. The algorithm designed in section~\ref{sec:algo} is also written for polynomials and in the easiest case can be identified with Horner algorithm of quick evaluation of a polynomial. An interesting by-product of the construction is the obtention of a variant of Horner's rule \cite{enwiki:1222412612} of quick simultaneous evaluation of a polynomial and its derivatives on some element.

\subsection{Towards Yang-Baxter equation and beyond}

Once the correspondence with two-dimensional statistical mechanics is established for two constant $p$ and $q$ numeration bases, the following new step made in section~\ref{sec:YB} consists in the observation that the elementary weights of the configurations of the models, i.e. the digits in suitable bases, satisfy the so-called Yang-Baxter equation. This corresponds to an associativity in three changes of numeration bases between bases $p$, $q$ and $r$. We have not been able to find any trace of this observation, at least in this language, in the literature despite of its simplicity. The Yang-Baxter equation in statistical mechanics has a long story, starting in some hidden way in Onsager's solution of the Ising model and culminating in Baxter's solution of the eight-vertex model (see \cite{BaxterBook} for more details) and is the basic brick on which rely almost all exact formulae in integrable models of statistical mechanics such as the six-vertex model for example. Discovering the presence of this Yang-Baxter equation in number theory may explain some existing exact formulae or produce new ones, as it is the case in statistical mechanics, even if, for us, it is too early to have a clear insight beyond this analogy. The Yang-Baxter equation also has numerous ramifications in various part of algebra (construction à la Drinfeld of Hopf algebras \cite{Drinfeld}, Poisson-Lie groups, Kac-Moody algebras) and it would be interesting to see in later works if the Yang-Baxter family constructed here can also be obtained by other means.

The Yang-Baxter equation in dimension two can be seen as the projection of a higher three-dimensional cube relation. This is indeed the case here: three constant bases $p$, $q$ and $r$ give rise to a three-dimensional lattice where oriented broken lines are related to mixed-radix bases. We see in section~\ref{sec:YB} that this results can even be pushed even further: any sequence of constant bases $(p_i)$ with a finite sequence of $d$ distinct integers $p_i$ can be interpreted as a $d$-dimensional hyper-cubic lattice with suitable mixed radix bases on oriented broken lines.

\subsection{An application in dynamical systems: the Furstenberg conjecture}

In order to emphasize the role that \mrbs may play in practice, we reformulate in section~\ref{sec:examples} results related to Furstenberg's $2x$-$3x$ conjecture \cite{Furstenberg_1967} in terms of Yang-Baxter equation for \mrbs and some more recent works on two-dimensional discrete Markov processes in statistical mechanics. The measures on $\setR/\setZ$ that are invariant under both maps $x\mapsto 2x$ and $x\mapsto 3x$ are conjectured to be either the Lebesgue measure or finite-support atomic measures. We present here how these two families of measures are related with \mrbs and the set-theoretic Yang-Baxter equations through a multi-layer construction. We also describe how some potential improvements can be imagined from this construction.

\section{\Mrbs and change of bases for integers and polynomials}
\label{sec:mixedradix}

\subsection{Definitions and examples}

\begin{defi}\label{def:mixedradix}
A \mrb on $\setN$ is a sequence $\gr{b}=(b_i)_{i\geq 1}$ of integer numbers strictly larger than $1$. We define the product sequence $\pi^\gr{b}=(\pi^\gr{b}_i)_{i\geq 0}$ as $\pi^\gr{b}_0=1$ and $\pi^\gr{b}_i= \prod_{j=1}^i b_i$ for $i\geq 1$.
\end{defi}
An alternative definition may consider first the sequence $(\pi_i)_{i\geq 0}$ such that $\pi_0=1$ and, for all $i\geq 0$, $\pi_i$ divides $\pi_{i+1}$; in this case, the sequence $b_i$ is defined as $b_i=\pi_{i}/\pi_{i-1}$.

A constant radix $p$ basis corresponds to the choice of a constant sequence $b_i=p$ and a geometric sequence $\pi_i=p^i$. 

We also consider a \emph{finite} \mrb as a finite sequence $\gr{b}=(b_i)_{1\leq i\leq K}$.

\begin{prop}\label{prop:mixedradix}
Let $\gr{b}$ be a \mrb. Then, for any integer $n\in\setN$, there exists a unique sequence $(a_i)_{i\geq 0}$ such that
\begin{enumerate}[(i)]
	\item for all $i\geq 0$, $0\leq a_i< b_{i+1}$;
	\item there exists $k_n\in\setN$ such that $a_i=0$ for $i\geq k_n$;
	\item the integer $n$ can be written as
	\begin{equation}\label{eq:mixedradixrep}
		n = \sum_{i=0}^{\infty} a_i \pi^{\gr{b}}_i = a_0 + a_1 b_1 +a_2 b_1b_2+\ldots + a_k b_1\ldots b_k +\ldots
	\end{equation}
\end{enumerate}
The sequence $(a_i)$ is called the \emph{sequence of digits} of $n$. We now recall quickly the fundamental decomposition of a number in a given basis. 
\end{prop}
We will often need the following result about finite \mrbs in the next sections.
\begin{prop}\label{prop:mixedradix:finite}
	Let $K\in\setN$. Let $\gr{b}=(b_i)_{1\leq i\leq K}$ be a \mrb.  Then, for any integer $0\leq n<\pi^\gr{b}_K$, there exists a unique sequence $(a_i)_{0\leq i<K}$ such that
	\begin{enumerate}[(i)]
		\item for all $0\leq i<K$, $0\leq a_i< b_{i+1}$;
		\item the integer $n$ can be written as
		\begin{equation}\label{eq:mixedradixrep:finite}
			n = \sum_{i=0}^{K-1} a_i \pi^{\gr{b}}_i 
		\end{equation}
	\end{enumerate}
	In particular, this induces a bijection $\{0,\ldots,\pi_K-1\}\to \prod_{i=0}^{K-1} \{0,\ldots,b_{i+1}-1\}$ which associates to $n$ its sequence of digits.
\end{prop}

The \mrbs are also called Cantor bases since they have been introduced formally by Cantor in 1869 \cite{Cantor1869} although some particular examples have been used by the Babylonians.

In all the \mrbs, addition and multiplication can be performed using the standard algorithm of addition or multiplication of digits with carries. Truncation of digit decompositions provides the standard quotient operations $\setZ \to \setZ/\pi^{\gr{b}}_k \setZ$ with the classical projective structure between these rings.

The most classical cases of numeration bases are the following ones:
\begin{itemize}
	\item the constant radix $p$ basis for any integer $p$ strictly larger than $1$: it corresponds to the case of a constant sequence $b_i=p$ and $\pi^\gr{b}_i=p^i$. Our numerical system is based on $p=10$ and computer arithmetic on $p=2$.
	
	\item the factorial basis: it corresponds to $b_i=i+1$ and $\pi^{\gr{b}}_i=(i+1)!$ and it induces a bijection called the Lehmer code between any integer $n\in\{0,1,\ldots,N!-1\}$ with the symmetric group $\Sym{N}$.
	
	\item the time basis with $b_1=b_2=60$ (numbers of seconds in one minute and number of minutes in one hour), $b_3=24$ (numbers of hours in a day), $b_4=7$ (number of days in one week).
\end{itemize}
These \mrbs considered alone appear as interesting curiosities. We will now consider the more interesting case of a  change of basis.

\subsection{Change of bases} 
\subsubsection{The naive picture}\label{sec:naivepic}
Given a given number $n$ and a constant radix $p$ basis, the digit decomposition can be obtained by performing successive Euclidean division by $p$ on the successive quotient while keeping the remainders as the wanted digits. A change of constant radix bases from $p$ to $q$ can then be naively performed in two steps: starting from a sequence of digits in the constant radix $p$ basis, one first compute the integer $n$ using \eqref{eq:mixedradixrep} and then the Euclidean division algorithm by $q$ started on this $n$ provides the digits in a second constant radix $q$ basis. However, such an algorithm uses a large number of multiplications and Euclidean divisions that are costly in practice and quickly requires the use of very large numbers $n$.

A second idea is to use a Horner-like algorithm (see below for the original Horner algorithm for polynomials) to obtain the first digits in a constant radix $q$ basis from the ones in the constant radix $p$ basis. Evaluation of $n$ from the digits in \eqref{eq:mixedradixrep} can be done using few products by using the following recursion: one first considers the digit $a_k$ with largest position, i.e. the last non-zero digit and define the following sequence $(m_i)_{0\leq i\leq k}$ defined by $m_k=a_k$ and $m_{i}=m_{i+1}p+a_i$ for $i<k$; the number $n$ then corresponds to $m_0$. The obtention of the first digit $a'_0$ in constant radix $q$ basis can be obtained by using the same recursion on $\setZ/q\setZ$ instead of $\setZ$. Working instead in $\setZ/q^k\setZ$ provides a number which can be further decomposed to obtain the first $k$ digits. This last classical algorithm of computation of $a'_0$ requires only a finite number of addition and multiplication to establish in $\setZ/q\setZ$ the table of the map $(m,a)\mapsto mp+a$, which is then used at each step.

The idea of the paper is to provide interesting interpretations of all the intermediate computations steps $(m_i)$ and to enhance them to a higher level.

\subsubsection{An elementary change of \mrb and a fundamental lemma.}

Let $\gr{b}$ be a \mrb (resp. finite \mrb of length $K$). For any $N\in\setN$ (resp. for $N=K$), the symmetric group $\Sym{N}$ acts on the set of \mrbs through \begin{equation}\label{eq:actionSymmetricGroup}
	\sigma(\gr{b})_i = b_{\sigma^{-1}(i)}
\end{equation}
for any $\sigma\in\Sym{N}$. The inductive structure of the symmetric groups is compatible with this action and induce an action of the whole group $\Sym{\infty}=\varinjlim \Sym{N}$ but we will not need this fact; it only makes easier some formulations.

A first interesting fact is that this action leaves invariant any constant radix $p$ basis and this may explain why the present construction has not been described before. It then becomes interesting only for non-trivial \mrbs and we now state the fundamental lemma of the present paper. To this end, we introduce the fundamental functions $\psi_{p,q}$ defined by:
\begin{equation}\label{eq:def:psi}
	\begin{split}
		\psi_{p,q}  : \{0,1,\ldots,p-1\}\times\{0,1,\ldots,q-1\} & \to \{0,1,\ldots,q-1\}\times\{0,1,\ldots,p-1\}
		\\
		(u,v) &\mapsto (\psi^{(1)}_{p,q}(u,v),\psi^{(2)}_{p,q}(u,v)) 
	\end{split}
\end{equation}
where the two elements $v'=\psi^{(1)}_{p,q}(u,v)$ and $u'=\psi^{(2)}_{p,q}(u,v)$ are the unique integers such that
\[
 u+pv = v'+qu'
\] 
One trivially observes that $\psi_{p,q}\circ\psi_{q,p}=\id$ and $\psi_{p,p}(u,v)=(v,u)$.

This can also be interpreted as a simple case of change of \mrb: starting from $(u,v)$, one computes the integer $n=u+pv\in \{0,1,\ldots,pq-1\}$ in the finite \mrb $(b_1,b_2)=(p,q)$ and then decompose $n=v'+qu'$ in the transposed finite \mrb $(b'_1,b'_2)=(q,p)$. The fundamental lemma below generalizes this fact to arbitrary \mrbs and arbitrary transpositions in a "local" way.

\begin{lemm}[change of basis under transposition for numbers]\label{lemm:fundanumbers}
	Let $N\in\setN\cup\{\infty\}$ and let $\gr{b}$ be a \mrb of length $N$. Let $0\leq k<N-1$ and let $\tau_{k,k+1}$ the transposition of the two elements $k$ and $k+1$. For any $0\leq n <\pi_{N}$ (with the convention $\pi_\infty=\infty$), the decomposition $(a'_i)_{0\leq i<N}$ of $n$ in the \mrb $\tau_{k,k+1}(\gr{b})$ is obtained from the decomposition $(a_i)_{0\leq i<N}$ in the \mrb $\gr{b}$ through, for all $i$,
	\begin{equation}
		a'_i = \begin{cases}
			   \psi_{b_k,b_{k+1}}^{(1)}(a_{k-1},a_k) & \text{if $i=k-1$}
			   \\
			   \psi_{b_k,b_{k+1}}^{(2)}(a_{k-1},a_k) & \text{if $i=k$}
			   \\
			   a_i & \text{else}
			\end{cases}
	\end{equation} 
\end{lemm}
\begin{proof}
Let $n=\sum_{0\leq i <N} a_i \pi^{\gr{b}}_i$. The series is always finite for a finite $n$ since the digits $a_i$ are zeroes for sufficiently larger $i$.  We thus have
\begin{align*}
	n = \sum_{0\leq i <k-1} a_i \pi^{\gr{b}}_i + \pi^{\gr{b}}_{k-1}(a_{k-1}+b_{k}a_k) + \sum_{k+1\leq i <N} a_i \pi^{\gr{b}}_i
\end{align*} 
with $(a_{k-1},a_{k})\in \{0,1,\ldots,b_k-1\}\times\{0,\ldots,b_{k+1}-1\}$. By construction, we have
\[
a_{k-1}+b_{k}a_k = \psi_{b_k,b_{k+1}}^{(1)}(a_{k-1},a_k) + b_{k+1} \psi_{b_k,b_{k+1}}^{(2)}(a_{k-1},a_k)
\] 
Let $\gr{b}'= \tau_{k,k+1}(\gr{b}) $. We observe that $\pi^{\gr{b}'}_{i} = \pi^\gr{b}_{i}$ if $i<k$ and $i\geq k+1$ and $\pi^{\gr{b}'}_k=\pi^{\gr{b}'}_{k-1}b_{k+1}$.
From the definition of the sequences $(a'_i)$ and $(\pi^{\gr{b}'}_i)$, we then have
\[
	n = \sum_{0\leq i <k-1} a'_i \pi^{\gr{b}'}_i + \pi^{\gr{b}'}_{k-1}(a'_{k-1}+b_{k+1}a'_k) + \sum_{k+1\leq i <N} a'_i \pi^{\gr{b}'}_i
	= \sum_{0\leq i < N} a'_i\pi^{\gr{b}'}_i
\]
and we conclude by unicity of the decomposition.
\end{proof}

The previous lemma has an interesting computational consequence since it identifies an elementary change of numeration basis in which only two digits are transformed and the other digits remain the same. The algorithmic consequences are developed in section~\ref{sec:algo}.

\subsection{The polynomial case}

A similar situation of change of \mrb occurs for polynomials. We consider here only the case of polynomials $\setK[X]$ over a field $\setK$ to make things simpler. Traditional homogeneous bases are given $((X-a)^n)_{n\in\setN}$ for any given $a\in\setK$. In this case, the coefficients in such a basis are related to derivatives through:
\begin{align*}
P &= \sum_{n\in\setN} p_n(a) (X-a)^n
\\
p_n(a) &= \frac{P^{(n)}(a)}{n!}
\end{align*}
and the change of basis is a triangular linear map on the coefficients $(p_n(a))$  obtained from Newton's binomial formula:
\[
p_k(b) = \sum_{n\geq k} \binom{n}{k} p_n(a) (b-a)^{n-k} 
\]
In particular, we have $p_0(b)=P(b)$ can be obtained using only a few multiplications through Horner's rule from $u_0$ in the backward  affine recursion ($d$ is the degree of $P$)
\begin{equation}\label{eq:hornerpoly}
\begin{cases} 
	u_d = p_d(a) & \text{(highest degree term)}
	\\
	u_{k-1} = (b-a) u_k + p_{k-1}(a) & \text{for $1\leq k\leq d$}
\end{cases}
\end{equation}
which is similar to the number case.

A \mrb for polynomials corresponds here to a sequence $\gr{a}=(a_k)_{1\leq k}$ in $\setK$ and, indeed, any polynomial can be written uniquely as
\begin{equation}\label{eq:polyn:mrbcoeffs}
P = \sum_{n\in\setN} p_n(\gr{a}) \pi_n^{\gr{a}}(X) = \sum_{n\in\setN} p_n(\gr{a})\prod_{k=1}^n (X-a_k)
\end{equation}
where the sequence $(p_n(\gr{a})$ is a sequence of complex numbers vanishing for sufficiently large $n$ and the product of monomials play the role of the sequence $(\pi_i)$ for the integers. 

We still have $p_0(\gr{a})=P(a_1)$ but the simple interpretation in terms of derivatives is lost for higher terms but is replaced by finite differences. We now introduce the finite difference operator $(\Delta_u P)(X)=(P(X)-P(u))/(X-u)$ on polynomials. If all the $(a_k)$ are distinct, it is an easy exercise to obtain the following relations:
\begin{subequations}\label{eq:coeffinterpretations}
	\begin{align}
		p_0(\gr{a}) &=	P(a_1) 
		\\
		p_1(\gr{a}) &=	\frac{P(a_2)-P(a_1)}{a_2-a_1}   
		=(\Delta_{a_1}P)(a_2)\\
		p_2(\gr{a}) &=	\frac{1}{a_3-a_2}\left(\frac{P(a_3)-P(a_1)}{a_3-a_1}-	\frac{P(a_2)-P(a_1)}{a_2-a_1}\right) = \Delta_{a_2}(\Delta_{a_1}P)(a_3) 
	\end{align}
\end{subequations}
and so on. If some of the $a_k$ coincide then limits of the previous expressions have to be considered in the l.h.s. and some terms have to be replaced by derivatives.

We now introduce an analogue $\phi_{a,b}$ of the functions $\psi_{p,q}$ in \eqref{eq:def:psi} through the following definition for any complex numbers $a$ and $b$:
\begin{equation}\label{eq:def:phi}
	\begin{split}
	\phi_{a,b} : \setK^2&\to \setK^2 \\
			\begin{pmatrix}u\\v\end{pmatrix} & \mapsto 
			\begin{pmatrix}
				\phi^{(1)}_{a,b}(u,v) \\
				\phi^{(2)}_{a,b}(u,v)
			\end{pmatrix}=
			\begin{pmatrix}
				1 & b-a \\ 0 &1
			\end{pmatrix}\begin{pmatrix} u\\v\end{pmatrix}
	\end{split}
\end{equation}
This corresponds to the change of basis of the same degree $1$ polynomial in the two \mrb with $(a_1,a_2)=(a,b)$ and $(a_1,a_2)=(b,a)$ by considering \[u+v(X-a) = (u+v(b-a)) + v (X-b) = v'+u'(X-b)\].

\begin{lemm}[change of basis under transposition for polynomials]\label{lemm:fundapolynomials}
	Let $\gr{a}=(a_i)_{1\leq i}$ be a \mrb of polynomials. For any given $k\in\setN^*$, let $\tau_{k,k+1}$ be the transposition of the two elements $k$ and $k+1$. The coefficients $(p'_n)_{n\in\setN}$ of $n$ in the \mrb $\tau_{k,k+1}(\gr{a})$ of a polynomial $P$ are obtained from the coefficients $(p_i)_{n\in\setN}$ in the \mrb $\gr{a}$ through 
	\begin{equation}
		p'_n = \begin{cases}
			\phi_{a_k,a_{k+1}}^{(1)}(p_{k-1},p_k) & \text{if $n=k-1$}
			\\
			\phi_{a_k,a_{k+1}}^{(2)}(p_{k-1},p_k) & \text{if $n=k$}
			\\
			p_n & \text{else}
		\end{cases}
	\end{equation} 
\end{lemm}
We do not repeat here the proof since it follows the same steps as for numbers. As for numbers, the introduction of \mrbs introduces elementary local operations for coefficients, which can be exploited to obtain interesting alternative representations of polynomials and alternative algorithms of change of basis.

\section{Horner-type algorithmic applications}
\label{sec:algo}

\subsection{\Mrbs between two constant radix bases}
In this section, we fix once and for all two integers $p$ and $q$ strictly larger than $1$ and consider the associated constant radix $p$- and $q$-bases defined by $b_k=p$ and $b'_k=q$ for all $k\in\setN$. Any integer $n\in\setN$ has unique decompositions
\[
n = \sum_{0\leq k} a_k p^k = \sum_{0\leq l} a'_l q^l
\]
with $0\leq a_k <p$ and $0\leq a'_l <q$ and we consider the algorithmic computation of the digits $(a'_l)$ from the digits $(a_k)$.

We will consider here only \emph{finite} \mrbs in order to have initialization for our algorithm but the length of the bases can be arbitrarily large so that any number can be treated by the algorithms below.

Converting a finite constant radix $p$ decomposition with size $l_1$ to a finite constant radix $q$ decomposition of length $l_2$ requires the introduction of the following set of \mrbs with extended length.
\begin{equation}\label{eq:interpmrb}
\ca{B}_{p,q}(l_1,l_2) = \left\{ \gr{b}\in \{p,q\}^{l_1+l_2} ; |\{1\leq i\leq l_1+l_2; b_i=p\}|=l_1  \right\}
\end{equation}
In an sequence $\gr{b}\in\ca{B}_{p,q}(l_1,l_2)$, the number $p$ appears $l_1$ times and $q$ appears $l_2$ times. The symmetric group $\Sym{l_1+l_2}$ acts as previously on $\ca{B}_{p,q}(l_1,l_2)$ and this action is transitive. All of these bases in $\ca{B}_{p,q}(l_1,l_2)$ allows one to decompose any number $n\in\{0,1,\ldots,p^{l_1}q^{l_2}-1\}$ in $l_1+l_2$ suitably chosen digits either in $\{0,1,\ldots,p-1\}$ or in $\{0,1,\ldots,q-1\}$.

Two particular \mrbs are extremal: the basis with $b^{\North\West}_i=q$ for $1\leq i\leq l_2$ and $b_i^{\North\West}=p$ for $i>l_2$ and the basis $b^{\South\East}_i=p$ for $1\leq i\leq l_1$ and $b^{\South\East}_i=q$ for $i>l_1$. We first embed the constant radix $p$ and $q$ bases in these two extremal \mrbs in the following trivial way.

\begin{lemm}\label{lemm:borderrep}
Any integer $n\in\{0,1,\ldots,p^{l_1}-1\}$ (resp. $n\in\{0,1,\ldots,q^{l_2}-1\}$) with a decomposition 
\[
n = \sum_{0\leq k<l_1} a_k p^k  \qquad \left(\text{resp. } n = \sum_{0\leq l<l_2} a'_l q^l \right)
\]
admits the following decomposition in $\gr{b}^{\South\East}$ (resp. in $\gr{b}^{\North\West}$):
\[
n = \sum_{0\leq i<l_1+l_2} \ov{a}_i \pi_i^{\South\East}  \qquad \left(\text{resp. } n = \sum_{0\leq i<l_2+l_2} \ov{a}'_i \pi_i^{\North\West}\right)
\]
with $\ov{a}_i=a_i$ for $i<l_1$ and $\ov{a}_i=0$ for $i\geq l_1$ and $\ov{a}'_i=a'_i$ for $i<l_2$ and $\ov{a}'_i=0$ for $i\geq l_2$.
\end{lemm} 

In particular, any number $0\leq n < \min(p^{l_1},q^{l_2})$ admits four basic decompositions, one in the constant radix $p$ basis, one in the constant radix $q$ basis, one in  $\gr{b}^{\South\East}$ and one in $\gr{b}^{\North\West}$. More generally, this number $n$ admits decompositions in all the \mrbs in $\ca{B}_{p,q}(l_1,l_2)$ and these decompositions are related by transformations under $\Sym{l_1+l_2}$. 

\begin{prop}\label{prop:changemrb:decomp}
	For all pair $(\gr{b},\gr{b}')$ of \mrbs in $\ca{B}_{p,q}(l_1,l_2)$, there exists a sequence of transformations $(t_k)_{1\leq k\leq K}$ of digits such that:
	\begin{enumerate}[(i)]
		\item for all $0\leq n <p^{l_1}q^{l_2}$, $t_{K}\circ t_{K_-1}\circ\ldots\circ t_1$ maps the decomposition of $n$ of in the  basis $\gr{b}$ to the decomposition of $n$ in the basis $\gr{b}'$
		\item all $1\leq k\leq K$, $t_k$ transforms only two digits and leaves the other unchanged.
	\end{enumerate}
\end{prop}
\begin{proof}
	The symmetric group $\Sym{l_1+l_2}$ acts transitively on $\ca{B}_{p,q}(l_1,l_2)$ and thus there exists $\sigma \in\Sym{l_1+l_2}$ such that $\sigma(\gr{b})=\gr{b}'$. Moreover, the permutation $\sigma$ can be decomposed as a product of transpositions $t_k=\tau_{i_{k},i_{k}+1}$. The result is thus achieved by using lemma~\ref{lemm:fundanumbers}.
\end{proof}

\subsection{A two-dimensional interpretation of the intermediate \mrbs.}

\subsubsection{Rectangles and paths}
We now provide the key result of the paper: a two-dimensional interpretation of proposition~\ref{prop:changemrb:decomp} together with algorithmic application. To this purpose, we start with description of the rectangular graphs that are required.

For any positive integers $l_1$ and $l_2$, we introduce the rectangle 
\[
	R(l_1,l_2)=\{0,1,\ldots,l_1\}\times \{0,1,\ldots,l_2\}\subset \setZ^2
\]
endowed with the unoriented graph structure defined by the set $E_h(l_1,l_2)$ of size $l_1(l_2+1)$ horizontal edges $\{(x,y),(x+1,y)\}$ and the set $E_v(l_1,l_2)$ of size $l_2(l_1+1)$ of vertical edges $\{(x,y),(x,y+1)\}$ inside $R(l_1,l_2)$. In order to simplify notations, an unoriented edge will be represented by its middle point in $\setZ\times (\setZ+1/2)$ for vertical edges and in $ (\setZ+1/2)\times\setZ$ for horizontal edges through:
\begin{align*}
	\{(k,l),(k+1,l) \} &\leftrightarrow (k+1/2,l)
	\\
	\{(k,l),(k,l+1) \} &\leftrightarrow (k,l+1/2)
\end{align*}

Any \mrb in $\ca{B}_{p,q}(l_1,l_2)$ defines an oriented path  $\gamma_{\gr{b}} : \{0,1,\ldots,l_1+l_2\} \to R(l_1,l_2)$ which joins $\gamma_{\gr{b}}(0)=(0,0)$ to $\gamma_{\gr{b}}(l_1+l_2)=(l_1,l_2)$ such that each value $p$ (resp. $q$) moves the path to the right (resp. to the top), i.e., for all $0\leq k<l_1+l_2$,
\[
\gamma_{\gr{b}}(k+1)-\gamma_{\gr{b}}(k) = \begin{cases}
	(1,0) & \text{if $b_{k+1}=p$}	
	\\
	(0,1) & \text{if $b_{k+1}=q$}	
\end{cases}
\]

\subsubsection{Decorations on rectangles and number decompositions}

\begin{defi}
Let $A_1$ and $A_2$ be two sets. A \emph{$(A_1,A_2)$-decoration} on a rectangle $R(l_1,l_2)$ is a pair of maps $\alpha=(\alpha_h,\alpha_v)$ with $\alpha_h: E_h(l_1,l_2) \to A_1$ and $\alpha_v:E_v(l_1,l_2) \to A_2$.

For any function $f : A_1\times A_2 \to A_2\times A_1$,
an \emph{$f$-compatible} decoration is a $(A_1,A_2)$-decoration $\alpha=(\alpha_h,\alpha_v)$ such that, for any $0\leq x<l_1$ and and $0\leq y<l_2$,
\begin{equation}\label{eq:deco:localcompat}
f\left(
 \alpha_h( x+1/2,y ), \alpha_v( x+1,y+1/2 )
\right)
\\
=
\left(
\alpha_v( x,y+1/2 ), \alpha_h( x+1/2,y+1 )
\right)
\end{equation}

We note $\gr{D}_{A_1,A_2}(f;l_1,l_2)$ the set of $f$-compatible $(A_1,A_2)$-decorations on a rectangle $R(l_1,l_2)$, or simply $\gr{D}(f;l_1,l_2)$ if $A_1$ and $A_2$ are clear from context or from $f$.
\end{defi}
In order to manipulate lighter notations, we will also write a decoration as a function $\alpha: E_h(l_1,l_2)\cup E_v(l_1,l_2)\to A_1\sqcup A_2$ such that, for any $e\in E_h(l_1,l_2)$, $\alpha(e)=\alpha_h(e) \in A_1$, and, for any $e\in E_v(l_1,l_2)$, $\alpha(e)=\alpha_v(e) \in A_2$.

The compatibility condition \eqref{eq:deco:localcompat} says that for any elementary one-by-one square of $R(l_1,l_2)$, the labels of the top (North) and left (West) edges are obtained by applying $f$ to the labels of the bottom (South) and right (East) edges. An example is provided in figure~\ref{fig:rectangleexample} for $A_1=\{0,1,2\}$, $A_2=\{0,1,2,3,4\}$ and $f=\psi_{3,5}$ defined in \eqref{eq:def:psi}.

\begin{figure}
	\begin{center}
		\begin{tikzpicture}[scale=1.25]
			\foreach \x in {0,1,...,5} {
				\draw (\x,0) -- (\x,4);
			}
			\foreach \y in {0,1,...,4} {
				\draw (0,\y) -- (5,\y);
			}
			\begin{scope}[shift={(0.2,-0.2)}]
				\draw[->,thick,red] (0,0)--(5,0)--(5,4);
				\node at (5,1) [red,right] {$\gr{b}^{\South\East}$};
			\end{scope}
			\begin{scope}[shift={(-0.2,0.2)}]
				\draw[->,thick,red] (0,0)--(0,4)--(5,4);
				\node at (0,3) [red,left] {$\gr{b}^{\North\West}$};
			\end{scope}
			\begin{scope}[shift={(0.05,-0.05)}]
				\draw[->,thick,blue] (0,0)--(2,0)--(2,1)--(3,1)--(3,2)--(4,2)--(4,4)--(5,4);
				\node at (3.,1.5) [right,blue] {$\gr{b}$};
			\end{scope}
			\begin{scope}[shift={(-0.05,0.05)}]
				\draw[->,thick,violet] (0,0)--(2,0)--(2,1)--(2,2)--(3,2)--(4,2)--(4,4)--(5,4);
				\node at (2.5,2) [above,violet] {$\gr{b'}$};
			\end{scope}
			\draw[->,thick] (2.8,1.2)-- node {$\psi_{3,5}$} (2.2,1.8) ;
			\foreach \x/\n in {0/2,1/1,2/0,3/2,4/2}  
				\node at (\x+0.5,0)  {$\mathbf{\n}$};
			\foreach \x/\n in {0/2,1/2,2/1,3/1,4/0}  
				\node at (\x+0.5,1) {$\mathbf{\n}$};
			\foreach \x/\n in {0/2,1/2,2/0,3/0,4/0}  
				\node at (\x+0.5,2) {$\mathbf{\n}$};
			\foreach \x/\n in {0/1,1/0,2/0,3/0,4/0}  
				\node at (\x+0.5,3) {$\mathbf{\n}$};
			\foreach \x/\n in {0/0,1/0,2/0,3/0,4/0}  
				\node at (\x+0.5,4) {$\mathbf{\n}$};
			\foreach \y/\n in {0/1,1/4,2/3,3/1}  
				\node at (0,\y+0.5) {$\mathit{\n}$};
			\foreach \y/\n in {0/3,1/4,2/2,3/0}  
				\node at (1,\y+0.5) {$\mathit{\n}$};
			\foreach \y/\n in {0/4,1/4,2/0,3/0}  
				\node at (2,\y+0.5) {$\mathit{\n}$};
			\foreach \y/\n in {0/3,1/1,2/0,3/0}  
				\node at (3,\y+0.5) {$\mathit{\n}$};
			\foreach \y/\n in {0/2,1/0,2/0,3/0}  
				\node at (4,\y+0.5) {$\mathit{\n}$};
			\foreach \y/\n in {0/0,1/0,2/0,3/0}  
				\node at (5,\y+0.5) {$\mathit{\n}$};	
		\end{tikzpicture}
	\end{center}
	\caption{\label{fig:rectangleexample}Two \mrbs $\gr{b}=(3,3,5,3,5,3,5,5,3)$ and $\gr{b}'=(3,3,5,5,3,3,5,5,3)$ for $(p,q)=(3,5)$ that differ by one transposition $\tau_{4,5}$ and the two extremal \mrbs $\gr{b}^{\South\East}=(3,3,3,3,3,5,5,5,5)$ and $\gr{b}^{\North\West}=(5,5,5,5,3,3,3,3,3)$ for lengths $(l_1,l_2)=(5,4)$. Along the edges are written the digits of the decomposition of $n=221$ in all the \mrbs of $\ca{B}_{3,5}(5,4)$: digits in bold on horizontal lines are base $3$ digits in $\{0,1,2\}$, digits in italic on vertical lines are base 5 digits in $\{0,1,2,3,4\}$.}
\end{figure}

\begin{theo}\label{theo:2Drep:numbers}
	For any integers $p$ and $q$ strictly larger than $1$, for any integers $l_1,l_2\in\setN$, there is a bijection 
	\begin{equation}\label{eq:numberbijection2D}
	 \begin{split}
	 	\Psi : \{0,1,2,\ldots,p^{l_1}q^{l_2}-1\} &\to \gr{D}(\psi_{p,q};l_1,l_2)
	 	\\
	 	n&\mapsto \Psi_n
	 \end{split}
 	\end{equation}
	such that, for any \mrb $\gr{b}\in\ca{B}_{p,q}(l_1,l_2)$, the decomposition of $n$ on $\gr{b}$ is given by the decorations along the path $\gamma_{\gr{b}}$ as follows:
	\begin{equation}\label{eq:integerpathrep}
		n = \sum_{0\leq i< l_1+l_2} \Psi_n( \{\gamma_\gr{b}(i),\gamma_\gr{b}(i+1)\}) \pi^{\gr{b}}_i
	\end{equation}
\end{theo}

Before providing the proof, we illustrate the theorem on the example described in figure~\ref{fig:rectangleexample}. We consider $(p,q)=(3,5)$ and $(l_1,l_2)=(5,4)$ and the number $n=221$ such that $n<3^5=243$ and $n<5^4=625$ (and thus $n\leq 3^5 5^4=151875$). One checks by hand that $n$ can be written as follows in the constant radix $3$ and $5$ bases:
\begin{align*}
	n =221_{10} &= \gr{2} + \gr{1}\cdot 3 + \gr{0}\cdot 3^2 + \gr{2}\cdot 3^3 + \gr{2}\cdot 3^4  = \gr{22012}_3
	\\
	&= \mathit{1} + \mathit{4}\cdot 5 + \mathit{3}\cdot 5^2 + \mathit{1}\cdot 5^3  = \mathit{1341}_5
\end{align*}
Using lemma~\ref{lemm:borderrep}, these two decompositions are completed by zeroes in the \mrb $\gr{b}^{\South\East}=(3,3,3,3,3,5,5,5,5)$ and $\gr{b}^{\North\West}=(5,5,5,5,3,3,3,3,3)$. The paths associated to $\gr{b}^{\South\East}$ (resp. $\gr{b}^{\North\West}$) follows the South (resp. West) edge of the rectangles and then the East (resp. North) edges, hence justifying the name a posteriori. The digits of the two decompositions above are then placed on the edges around the rectangle $R(5,4)$. For clarity, we keep the graphical distinction between bold and italic digits. We can consider any intermediate \mrb such as $\gr{b}=(3,3,5,3,5,3,5,5,3)$ or $\gr{b}'=(3,3,5,5,3,3,5,5,3)$. These two \mrbs are related by $\gr{b'}=\tau_{4,5}(\gr{b})$. In these \mrbs, $n=221$ can be decomposed as:
\begin{align*}
	n &= \gr{2} + \gr{1}\cdot 3 +\mathit{4}\cdot 3^2 + \gr{1} \cdot 3^2\cdot 5 + \mathit{1} \cdot 3^3\cdot 5 = \mathit{1}\gr{1}\mathit{4}\gr{1}\gr{2}_\gr{b}
	\\
	&=\gr{2} + \gr{1}\cdot 3 +\mathit{4}\cdot 3^2 + \mathit{4} \cdot 3^2\cdot 5 + \gr{0} \cdot 3^2\cdot 5^2 = \mathit{4}\mathit{4}\gr{1}\gr{2}_{\gr{b}'} = \gr{0}\mathit{4}\mathit{4}\gr{1}\gr{2}_{\gr{b}'}
\end{align*}
The first three digits in both bases coincide since the paths $\gamma_{\gr{b}}$ and $\gamma_{\gr{b}'}$ coincide along their first three increments. The next two digits differ through $(\mathit{4},\gr{0})=\psi_{3,5}(\gr{1},\mathit{1})$ since both paths are separated by an elementary square.

\begin{proof}[Proof of theorem~\ref{theo:2Drep:numbers}]
	The proof relies on the fact the decoration $\Psi_n$ is determined by a the single knowledge of the decomposition of $n$ in the basis $\gr{b}^{\South\East}$. The same lemma is at the origin of the algorithms presented below.
	
	\begin{lemm}\label{lemm:filling}
		Let $\alpha$ be a $f$-compatible decoration on a rectangle $R(l_1,l_2)$. Let $\gamma_{\gr{b}}$ be an oriented path from $(0,0)$ to $(l_1,l_2)$, which divides $R(l_1,l_2)$ into a North-West part $R^{\North\West}_\gamma(l_1,l_2)$ and $R^{\South\East}_\gamma(l_1,l_2)$. Then, $\alpha$ is entirely determined on $R^{\North\West}_\gamma(l_1,l_2)$ from its restriction along $\gamma$. If, furthermore, $f$ is invertible, then $\alpha$ is also determined on the whole rectangle $R(l_1,l_2)$.
	\end{lemm}
	\begin{proof}
		This is a direct consequence of the transitive action of $\Sym{l_1+l_2}$. A path in $R^{\North\West}_\gamma(l_1,l_2)$ (resp.  $R^{\South\East}_\gamma(l_1,l_2)$) is obtained from $\gamma_{\gr{b}}$ by successive transpositions $\tau_{i,i+1}$ in $\Sym{l_1+l_2}$ that map a pair $(p,q)$ (resp. $(q,p)$) to a pair $(q,p)$ (resp. $(p,q)$). This corresponds to mapping consecutive  horizontal then vertical South-East increments to the corresponding West and North consecutive increments around an elementary square. For each of these transpositions, the West and North decorations are obtained from the South and East decorations by applying $f$ as in \eqref{eq:deco:localcompat}.
		
		If $f$ is invertible, then \eqref{eq:deco:localcompat} can be written as
		\[
		\left(
		\alpha_h( x+1/2,y ), \alpha_v( x+1,y+1/2 )
		\right)
		=f^{-1}
		\left(
		\alpha_v( x,y+1/2 ), \alpha_h( x+1/2,y+1 )
		\right)
		\]
		which corresponds to a switch between $p$ and $q$, hence a flip of the rectangle $R(l_1,l_2)$ to $R(l_2,l_1)$ around the first diagonal and reversed roles between West and South and between North and East.
	\end{proof}

	The bijection $\Psi$ is then obtained by first decomposing $n$ along the \mrb $\gr{b}^{\South\East}$ basis and then a generation of the whole decoration $\Psi_n\in\gr{D}(\psi_{p,q};l_1,l_2)$ by the previous lemma. Formula~\ref{eq:integerpathrep} is then obtained by applying, for each transposition in the previous lemma, the construction of lemma~\ref{lemm:fundanumbers} and proposition~\ref{prop:changemrb:decomp}. The inverse bijection $\Psi^{-1}$ is then obtained directly from \eqref{eq:integerpathrep}.
\end{proof}

\subsection{The polynomial version of the two-dimensional representation}
For polynomials, the same type of results as theorem~\ref{theo:2Drep:numbers} holds \emph{mutatis mutandis}. The set $\ca{B}_{p,q}(l_1,l_2)$ is replaced by replacing the integers $p$ and $q$ by two complex numbers $x$ and $y$. The same geometric interpretation of an element $\gr{a}\in\ca{B}_{p,q}(l_1,l_2)$ as an oriented path in $R(l_1,l_2)$ still holds. Decorations are now functions from the edges of $R(l_1,l_2)$ to $\setK$. Theorem~\ref{theo:2Drep:numbers} is now replaced by the following theorem proved in the same way.

\begin{theo}\label{theo:2Drep:polynomial}
	For any integers $l_1,l_2\in\setN$, there is a bijection 
	\begin{equation}\label{eq:polynomialbijection2D}
		\begin{split}
			\Phi : \setK[X]_{l_1+l_2} &\to \gr{D}_{\setK,\setK}(\phi_{x,y};l_1,l_2)
			\\
			P &\mapsto \Phi_P
		\end{split}
	\end{equation}
	such that, for any \mrb $\gr{a}\in\ca{B}_{x,y}(l_1,l_2)$, the decomposition of $P$ on $\gr{a}$ is given by the decorations along the path $\gamma_{\gr{a}}$ as follows:
	\begin{equation}\label{eq:polynomialpathrep}
		P = \sum_{0\leq n< l_1+l_2} \Phi_P( \{\gamma_\gr{a}(n),\gamma_\gr{a}(n+1)\}) \prod_{i=1}^n (X-a_i)
	\end{equation}
\end{theo}

We now present an example to illustrate the construction with $x=0$ and $y$ and consider a polynomial $P$ of degree $3$.
\[
P = p_0+p_1 X + p_2 X^2 +p_3 X^3
\]
The coefficients $(p_i)$ are written on the South side of $R(4,4)$ and completed by $0$ along the East side. Applying successive transpositions as in lemma~\ref{lemm:filling} and using now the triangular linear map $\phi_{x,y}$ provides the whole decoration as illustrated in figure~\ref{fig:squareexample:polynomial}. The triangular structure of $\phi_{x,y}$ is obvious on the figure: on each elementary square, both North and East values are equal. One recognize on the West size the sequence $P^{(k)}(y)/k!$ of successive derivatives of $P$ evaluated at $y$, hence the change of basis.

By considering the path of length $4$ associated to $\gr{a}=(y,0,0,y)$, one immediately reads on the figure the following decomposition in the \mrb:
\begin{equation}
	P = P(y) + (p_1+p_2y+p_3y^2) (X-y) + (p_2+p_3y) (X-y)X + p_3  (X-y)X^2
\end{equation}

\begin{figure}
	\begin{center}
	\begin{tikzpicture}[xscale=2.1]
		\foreach \x in {0,1,...,4} {
			\draw (\x,0) -- (\x,4);
		}
		\foreach \y in {0,1,...,4} {
			\draw (0,\y) -- (4,\y);
		}
		\begin{scope}[shift={(0.2,-0.2)}]
			\draw[->,thick,red] (0,0)--(4,0)--(4,4);
			\node at (3,0) [red,below] {$\gr{\pi}^{\South\East}= (X^k)_{k}$};
		\end{scope}
		\begin{scope}[shift={(-0.2,0.2)}]
			\draw[->,thick,red] (0,0)--(0,4)--(4,4);
			\node at (1,4) [red,above] {$\gr{\pi}^{\North\West}= ((X-y)^k)_k$};
		\end{scope}
		\begin{scope}[shift={(-0.025,0.05)}]
			\draw[->,thick,blue] (0,0)--(0,1)--(2,1)--(2,2);
		\end{scope}
		\node at (0.5,0) {{\footnotesize $p_0$}};
		\node at (1.5,0) {{\footnotesize $p_1$}};
		\node at (2.5,0) {{\footnotesize $p_2$}};
		\node at (3.5,0) {{\footnotesize $p_3$}};
		\node at (0.5,1) {{\footnotesize $p_1+p_2y +p_3y^2$}};
		\node at (1.5,1) {{\footnotesize $p_2+p_3y$}};
		\node at (2.5,1) {{\footnotesize $p_3$}};
		\node at (3.5,1) {{\footnotesize $0$}};
		\node at (0.5,2) {{\footnotesize $p_2+2p_3y$}};
		\node at (1.5,2) {{\footnotesize $p_3$}};
		\node at (2.5,2) {{\footnotesize $0$}};
		\node at (3.5,2) {{\footnotesize $0$}};
		\node at (0.5,3) {{\footnotesize $p_3$}};
		\node at (1.5,3) {{\footnotesize $0$}};
		\node at (2.5,3) {{\footnotesize $0$}};
		\node at (3.5,3) {{\footnotesize $0$}};
		\node at (0.5,4) {{\footnotesize $0$}};
		\node at (1.5,4) {{\footnotesize $0$}};
		\node at (2.5,4) {{\footnotesize $0$}};
		\node at (3.5,4) {{\footnotesize $0$}};
		\node at (4,0.5) {{\footnotesize $0$}};
		\node at (4,1.5) {{\footnotesize $0$}};
		\node at (4,2.5) {{\footnotesize $0$}};
		\node at (4,3.5) {{\footnotesize $0$}};
		
		\node at (3,0.5) {{\footnotesize $p_3$}};
		\node at (3,1.5) {{\footnotesize $0$}};
		\node at (3,2.5) {{\footnotesize $0$}};
		\node at (3,3.5) {{\footnotesize $0$}};
		
		\node at (2,0.5) {{\footnotesize $p_2+p_3y$}};
		\node at (2,1.5) {{\footnotesize $p_3$}};
		\node at (2,2.5) {{\footnotesize $0$}};
		\node at (2,3.5) {{\footnotesize $0$}};
		
		\node at (1,0.5) {{\footnotesize $p_1+p_2y+p_3y^2$}};
		\node at (1,1.5) {{\footnotesize $p_2+2p_3y$}};
		\node at (1,2.5) {{\footnotesize $p_3$}};
		\node at (1,3.5) {{\footnotesize $0$}};
		
		\node at (0,0.5) [left] {{\footnotesize$p_0+p_1y+p_2y^2+p_3y^3$}};
		\node at (0,1.5) [left]{{\footnotesize $p_1+2p_2y+3p_3y^2$}};
		\node at (0,2.5) [left]{{\footnotesize $p_2+3p_3y$}};
		\node at (0,3.5) [left]{{\footnotesize $p_3$}};
		
		\draw[dashed] (4,0)--(0,4);
	\end{tikzpicture}
	\end{center}
	\caption{\label{fig:squareexample:polynomial}Decorations associated to a change of basis from $x=0$ to arbitrary $y\in\setK$ of a polynomial of degree $3$. Due to the graded structure of polynomials, only zeroes appear above the dashed line.}
\end{figure}
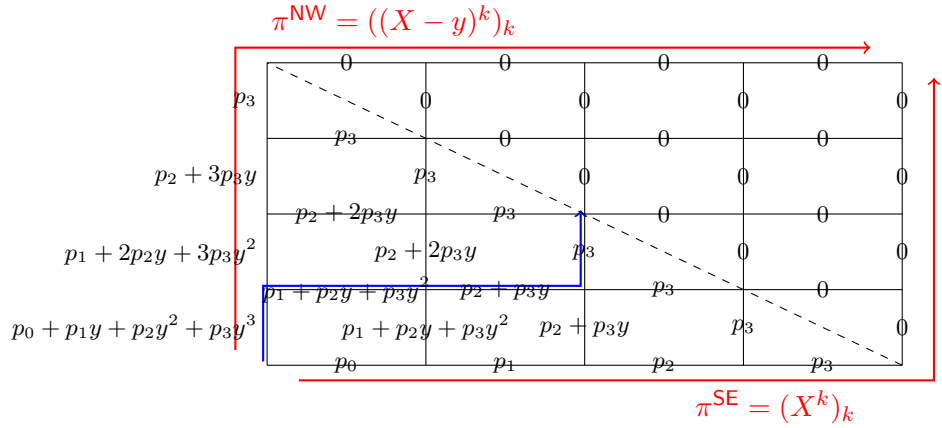

\subsection{Algorithmic content of the two-dimensional representation}
\subsubsection{The first line: Horner's rule}

The careful reader immediately recognize on the vertical edges of the first bottom line of figure~\ref{fig:squareexample:polynomial} the intermediate steps of Horner's rule of quick evaluation of a polynomial described in \eqref{eq:hornerpoly}. Indeed, one recognizes in this recursion the formula $u_{k-1}=\phi^{(1)}_{0,b-a}(p_{k-1}(a),u_k)$.

For numbers, the same is true and corresponds to the naive algorithm described in section~\ref{sec:naivepic}: the numbers on the vertical edges of the first bottom line of $R(l_1,l_2)$ corresponds to the recursion:
\[
u_{k-1} = \phi^{(1)}_{p,q}(a_{k-1},u_k) = a_{k-1}+p u_k \mod q
\]

The novelty here is the introduction of the full function $\phi$ and $\psi$ and not only the first components $\phi^{(1)}$ and $\psi^{(1)}$. This allows to close on the North the first bottom line of a rectangle and to start again the same Horner's rule on the second line of a rectangle. Moreover, lemma~\ref{lemm:filling} can be interpreted as a Markov property so that many computations can be parallelized. We present here algorithmic variants of the filling rules of lemma~\ref{lemm:filling}.

\subsubsection{Generation of a decoration}

In the following algorithms, we store $\alpha_h$ and $\alpha_v$ on $R(l_1,l_2)$ as two-dimensional arrays using the following the following shortcut notations:
\begin{align*}
	\alpha_h( i,j ) &:= \alpha_h( \{(i,j),(i+1,j)\} )   & 0\leq i<l_1, 0\leq j\leq l_2 \\
	\alpha_v( i,j ) &:= \alpha_h( \{(i,j),(i,j+1)\} )   & 0\leq i\leq l_1, 0\leq j< l_2 
\end{align*}
i.e. we remove the term $1/2$ used previously to indicate middles of edges in order to stick to computer science integer notations of arrays
The algorithm \ref{algo:decoration} generates the whole decoration associated to a South West boundary condition as in lemma~\ref{lemm:filling}. In this algorithms, both loops $1$ and $2$ can be interchanged. In the present case, all the columns are enumerate from right to left (loop 1) and, for each column, decorations are filled from bottom to top (loop $2$). One may switch both loops and fill lines from right to left, one after the other from bottom to top. One may also imagine hook iterations by filling the South East square and then fill in a parallel way the bottom line and the right column and then start again in the rectangle $[0,l_1]\times[1,l_2]$ of size $(l_1-1,l_2-1)$.

\begin{algorithm}
	\caption{General change of extremal basis in $\ca{B}_{p,q}(l_1,l_2)$\label{algo:decoration}}
	\begin{algorithmic}
		\Require $f: A\times B\to B\times A$
		\Require $x=(x_0,x_1,\ldots,x_{l_1-1},x'_0,\ldots,x'_{l_2-1}) \in A^{l_1}\times B^{l_2}$
		\Ensure 2D Arrays $\alpha_h$ and $\alpha_v$ of sizes $l_1(l_2+1)$ and $(l_1+1)l_2$.
		\State $\alpha_h(\bullet,0) \gets x_\bullet$		\Comment{South line init.}
		\State $\alpha_v(l_1,\bullet) \gets x'_\bullet$		\Comment{East line init.}
		\For{$i=l_1-1,\ldots,0$} \Comment{Loop 1}
		\For{$j=0,\ldots,l_2-1$}\Comment{Loop 2}
		\State $(\alpha_v(i,j),\alpha_h(i,j+1)) \gets f( \alpha_h(i,j), \alpha_v(i+1,j) )$
		\EndFor
		\EndFor
	\end{algorithmic}
\end{algorithm}

This algorithm works for both numeration bases and polynomials and provides two arrays on which one can read the decomposition in \emph{any} basis of $\ca{B}_{\bullet,\bullet}(l_1,l_2)$ provided we know the decomposition along an arbitrary given \mrb in $\ca{B}_{\bullet,\bullet}(l_1,l_2)$. However, in most cases, the input is only the particular case of lemma~\ref{lemm:borderrep}, in which the South line is provided and the East line is filled with $0$; moreover, in most cases, one is only interested in the content of the West line and one discards both the interior of the decoration and the North side. For these particular cases, we provide simplified algorithms below.

The complexity of algorithm~\ref{algo:decoration} is linear in the number of internal elementary squares and is thus given by $O(l_1l_2)$ where, which corresponds also the number of computed values. This is achieved using the fact that a right-to-left and bottom-to-top order of the squares allows to go only once through each of them and call the function $f$ only once per square. One also observes that the two for loops in this algorithm can be permuted by computing lines first, from bottom to top. In the concrete coding of this algorithm, further enhancements can be made to accelerate the program. For example, using parallelism, filling the $(i-1)$-th column from bottom to top can be started before the filling the $i$-th column is finished. 

\subsubsection{Enhanced Horner's rule for derivatives of polynomials}

Various Horner's rule exist in the literature of numeric computations to evaluate simultaneously a polynomial and its first derivative at the same point $y$ or to perform the complete Euclidean division of a polynomial $X-y$. All these computations are in fact much older than Horner's work but only the term "Horner's rule" survived in the literature (see \cite{enwiki:1222412612} and the associated discussion for an interesting point of view). All these variants correspond more or less to the following algorithm.

Let $P$ be a polynomial of degree $d$, $y$ be a given number and $k\geq 1$ a fixed integer. The quantity $(P(y),P'(y),\ldots,P^{(k-1)}(y))$ can be evaluated using the following algorithm within a rectangle of size $R(d+1,k)$ and more precisely in the lower left triangle $\{(i,j)\in R(d+1,k) ; i+j\leq d+1\}$. We consider the \mrbs in $\ca{B}_{0,y}(d+1,k)$ and simplify algorithm \ref{algo:decoration} to explore only the previous triangle (see figure~\ref{fig:squareexample:polynomial}) and reduce the use of memory space.

\begin{algorithm}
\caption{Simultaneous evaluation of the derivatives of a polynomial\label{algo:simultderivates}}
\begin{algorithmic}
	\Require $P=\sum_{k=0}^d p_k X^k$, $y\in\setK$
	\Ensure $u=(u_0,\ldots,u_{k-1})=(P(y),P'(y),\ldots,P^{(k-1)}(y))\in \setK^k$
	\State $(u_0,\ldots,u_{k-1}) \gets (p_d,0,0,\ldots,0)\in \setK^k$	\Comment{second to last vertical line}
	\For{$k$ from $d-1$ to $0$}
		\State $v\gets p_k$	\Comment{Value on the running horizontal segment}
		\For{$i$ from $0$ to $k-1$} \Comment{zeroes above the dashed line in fig.~\ref{fig:squareexample:polynomial}}
			\State $(u_i,v)\gets (v+u_i y,u_i)$ \Comment{corresponds to $\phi_{0,y}$}
		\EndFor
	\EndFor
	\State \Return $u$
\end{algorithmic}
\end{algorithm}

The complexity of algorithm \eqref{algo:simultderivates} is $O(k d)$: it is linear in both the degree $d$ of $P$ and the number $k$ of derivatives. This is similar to the application of $k$ times Horner's rule (one for each derivative). However, this algorithm is expected to run a bit faster since the array of coefficients of $P$ is browsed only once in the loop over $k$ and each coefficient is processed in the inner loop over $i$. 

A last advantage of this algorithm is that all the intermediate quantities in the computations have an interpretation similar to \eqref{eq:coeffinterpretations} can be stored and reused if needed.

\subsubsection{Conversion between constant radix $p$ to $q$ bases}

In the context of numeration basis, we start with a given decomposition $ n = \sum_{i=0}^{k-1} a_i p^i$ with $0\leq a_i<p$, $a_{k-1}\neq 0$. To obtain the full decomposition in the constant radix $q$ basis $n=\sum_{i=0}^{l-1} a'_i q^i$. Since the maximal value of $n$ is $p^k$, we need $q^l\geq p^k$ and thus at least $l= \lceil k\log p/\log q \rceil$. In the rectangle $R(k,l)$, it is clear that only the lower left triangle is filled. 

In the present case, the map $\psi_{p,q}$ is a function on a finite set of cardinal $pq$ that will be called a large number of times. Its definition include Euclidean division by $q$ that may potentially be a bit slow: it can be avoided by precomputing the table of all possible values of $\psi_{p,q}$ in a two-dimensional array of size $pq$.

\begin{algorithm}
	\caption{Conversion from constant radix $p$ basis to the constant radix $q$ basis\label{algo:changebasis}}
	\begin{algorithmic}
		\Require $p\geq 2$, $q\geq 2$, $k\geq 1$, $n=\sum_{i=0}^{k-1} a_i p^i$ with $0\leq a_i<p$
		\Ensure $(a'_0,\ldots,a'_{l-1})$ with $n=\sum_{i=0}^{l-1} a'_i q^i$ with $0\leq a'_i<q$
		
		\For{$i=0,\ldots,p-1$}	\Comment{Precomputation of the $\psi_{p,q}$ table}
			\For{$j=0,\ldots,q-1$} \Comment{with $pq$ Euclidean divisions}
				\State $\psi^{(1)}(i,j) \gets i+p j \mod q$
				\State $\psi^{(2)}(i,j) \gets (i+p j) / q$ 
			\EndFor
		\EndFor
		\State $\alpha \gets \log(p)/\log(q)$
		\State $a'\gets \emptyset$
		\For{$m$ from $k-1$ to $0$}
		\State $v\gets a_m$	\Comment{Horizontal segment running upwards}
		\State $h \gets \lceil (k-m)\alpha\rceil$ \Comment{Minimum height with only zeroes above}
		\State $a' \gets (a',0,...,0)$ to fill up to height $h$ \Comment{Zero default value above $h$}
		\For{$i$ from $0$ to $h-1$} 
		\State $(a'_i,v) \gets (\psi^{(1)}(v,a'_i), \psi^{(2)}(v,a'_i))$ \Comment{no division computed here !}
		\EndFor
		\EndFor
		\State \Return $a'$
	\end{algorithmic}
\end{algorithm}

It is interesting to remark that, once the table of values of $\psi_{p,q}$, no other arithmetic addition or product is required to fill the decorations on the triangle. The only non-trivial arithmetic operations are the ones in the line of algorithm \ref{algo:changebasis} that evaluates the current height $h$ but it can be also be avoided by transforming the inner loop into a while loop to describe the addition of new vertical segments on the columns when moving one column to the left.

An example of the intermediate values produced by algorithm~\ref{algo:changebasis} is provided on figure~\ref{fig:examplealgobasis}. 

\begin{figure}
\begin{center}
	\input{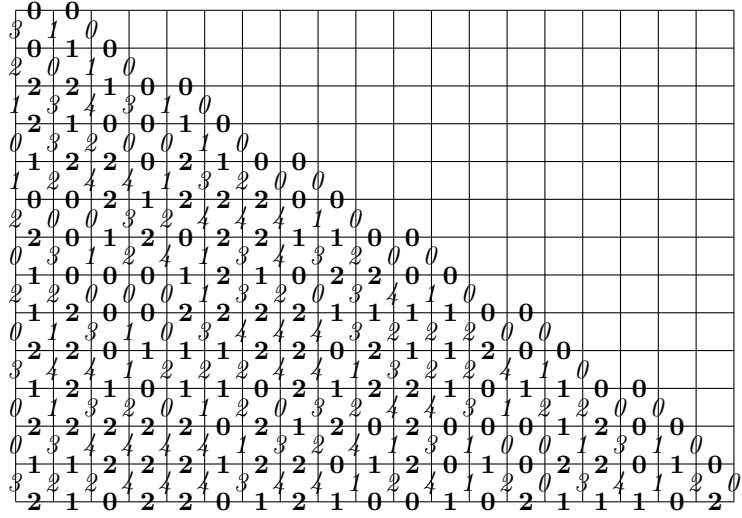}
\end{center}
\caption{\label{fig:examplealgobasis}Algorithm~\ref{algo:changebasis} running on an example with $(p,q)=(3,5)$ and $n=840397253$. Only values computed during the algorithm are present. The input is the lower South line and the output is the left West line. For each vertical strip, the variable $v$ in this algorithm corresponds successively to the values in $\{\gr{0},\gr{1},\gr{2}\}$ associated to the horizontal segements from bottom to top.}
\end{figure}

In terms of complexity, given a number $n$, it has around $\log(n)/\log(p)$ digits in the constant radix $p$ basis and around  $\log(n)/\log(q)$ digits in the constant radix $q$ basis. These numbers of digits are the lengths of the South and West segments in figure~\ref{fig:examplealgobasis}. The total size of the two loops in algorithm~\ref{algo:changebasis} is equal to the number of elementary squares to fill in figure~\ref{fig:examplealgobasis} and is thus equal to $\log(n)^2/(2\log(p)\log(q))$. The complexity of this algorithm is thus equal to \[O((\log(n))^2)\] In terms of computation times, this complexity hides two other advantages. First, the precomputation step of the $\psi_{p,q}$ table implies that the computation of $(a'_i,v)$ in the inner loop requires only memory access and no additional arithmetic operation and is thus really fast. Second, as in the case of polynomials discussed before, parallelism can be added to start the filling $(m-1)$-th column before the end of the computation of the $m$-th column. Finally, if needed, lines and columns can be permuted in the loops if one needs a quicker access to the first digits in the constant radix $q$ basis.

\subsection{Chinese remainder theorem and rotations}

We have seen that the maps $\psi_{p,q} : \{0,\ldots,p-1\} \times \{0,\ldots,q-1\} \to \{0,\ldots,q-1\} \times \{0,\ldots,p-1\}$ are invertible and correspond to shift of decorations from the South and East edges to the West and North edges. The inverse map provides a shift of decorations from West and North edges to South and East edges. One may wonder if similar maps exists along the other diagonal.

\begin{prop}
	Let $p$ and $q$ be two integers strictly larger than $1$ such that $\gcd(p,q)=1$. Then there exists a unique map $\psi_{q,p}^{\nearrow}: \{0,1,\ldots,q-1\}\times \{0,1,\ldots,p-1\} \to \{0,1,\ldots,p-1\}\times \{0,1,\ldots,q-1\}$ such that
	\[
	(u_{\West},u_{\North}) = \psi_{p,q}( u_{\South},u_\East) \qquad \Longleftrightarrow \qquad (u_\North,u_\East)=\psi_{q,p}^{\nearrow}(u_\West,u_{\South})
	\]
\end{prop}
\begin{proof}
	For any $(u_\West,u_\South)\in \{0,1,\ldots,q-1\}\times \{0,1,\ldots,p-1\}$, the Chinese remainder theorem states there exists a unique $n\in\{0,1,\ldots,pq-1\}$ such that $n \equiv u_\West \mod q$ and $n\equiv u_\South \mod p$. The integer can then be decomposed in the two \mrbs $(p,q)$ and $(q,p)$ along
	\[
		n = u_\West +q v' = u_\South + p u'
	\] with $0\leq v'<p$ and $0\leq u'<q$. We now define $\psi^{\nearrow}_{q,p}(u,v)=(v',u')$. Using the same decomposition, we also have  $(u_\West,v')=\psi_{p,q}(u_\South,u')$ and one checks directly the equivalence and unicity.
\end{proof}
Consequently, given a rectangle $R(l_1,l_2)$, the full decoration can be constructed when $\gcd(p,q)=1$ from the decorations on the left and bottom sides using $\psi^{\nearrow}_{q,p}$ and all the algorithms above can be adapted in order to produce all the decompositions along \mrbs in $\ca{B}_{p,q}(l_1,l_2)$ of the unique integer $0\leq n<p^{l_1}q^{l_2}$ such that 
	\begin{align*}
		n&\equiv \sum_{i=0}^{l_1-1} a_i p^i \mod p^{l_1} 
		&
		n&\equiv \sum_{j=0}^{l_2-1} a'_j q^j\mod q^{l_2} 
	\end{align*}

\section{Yang-Baxter equation and change of \mrbs}

\label{sec:YB}
In the previous section~\ref{sec:algo}, we have explored the algorithmic aspects of the conversion between \emph{two} homogeneous bases, both in the polynomial case and in the number case. We now consider the case of three homogeneous bases and establish a relation with the set-theoretic Yang-Baxter equation. We will see that, in the polynomial case, we recover classical "braiding" properties of triangular matrices whereas, for integers, a new class of integrable systems of statistical mechanics emerges.

\subsection{An overview of quantum and set-theoretic Yang-Baxter equation}

In order to formulate the quantum and set-theoretical Yang-Baxter equations, we need an additional notation. Let $A$ and $B$ be two sets and $F$ be a function $A\times B\to B\times A$. For any sequence of spaces $(E_k)_{1\leq k\leq n}$, we define, for all $1\leq i\leq n$ such that $E_i=A$ and $E_{i+1}=B$, the function $\iota_i(F)$ given by
\begin{align*}
	\iota_i(F) : \prod_{i=1}^n E_k & \to E_1\timesdots E_{i-1}\times E_{i+1}\times E_i \times E_{i+2}\timesdots E_n \\
	(x_1,\ldots,x_k)  &\mapsto (x_1,\ldots,x_i,F(x_i,x_{i+1}),x_{i+2},\ldots,x_n)
\end{align*}
This definition is trivially extended to the case of tensor products instead of cartesian products.

\begin{defi}[quantum Yang-Baxter equation]
	Let $V_1$, $V_2$ and $V_3$ be three vector spaces. Three linear maps \begin{align*}
		R_{12} : V_1\otimes V_2 &\to V_2\otimes V_1 \\
		R_{13} : V_1\otimes V_3 &\to V_3\otimes V_1 \\
		R_{23} : V_2\otimes V_3 &\to V_3\otimes V_2
	\end{align*}
	satisfy the quantum Yang-Baxter equation if and only if, the following equation holds on $V_1\otimes V_2\otimes V_3$ 
	\begin{equation}\label{eq:quantumYB}
		\iota_1(R_{23})\circ \iota_2(R_{13}) \circ \iota_1(R_{12})
		=
		\iota_2(R_{12})\circ \iota_1(R_{13}) \circ \iota_2(R_{23}) 
	\end{equation}
\end{defi}
Following the definition of the embeddings $\iota_k$, one observes that the codomain of both sides of \eqref{eq:quantumYB} is $V_3\otimes V_2\otimes V_1$.

In practice, the definition with $V_1=V_2=V_3$ and $R_{12}=R_{13}=R_{23}$ is often written but, for many applications (transfer matrices, Baxter $Q$-operators, higher spin models), it is more convenient to consider the case of varying spaces $V_i$ and/or three matrices $R_{ij}$ that act in different ways on each pair of spaces. For example the three matrices $R_{ij}$ may depend on a parameter $R_{ij}= R(u_i-u_j)$ or the spaces $V_i$ can be obtained as various representation spaces of an underlying algebra $\ca{A}$ endowed with an abstract element $\ca{R}$ in $\ca{A}\otimes\ca{A}$ (or in a completion of $\ca{A}\otimes\ca{A}$), as it is the case in the construction of $R$-matrices through Drinfeld quasi-triangular Hopf algebras.

\begin{defi}[set-theoretic Yang-Baxter equation]
	Let $S_1$, $S_2$ and $S_3$ be three sets. Three maps \begin{align*}
		R_{12} : S_1\times S_2 &\to S_2\times S_1 \\
		R_{13} : S_1\times S_3 &\to S_3\times S_1 \\
		R_{23} : S_2\times S_3 &\to S_3\times S_2
	\end{align*}
	satisfy the \emph{set-theoretic Yang-Baxter equation} if and only if \eqref{eq:quantumYB} holds on $S_1\times S_2\times S_3$ (cartesian product instead of tensor products)
\end{defi}
The set-theoretic Yang-Baxter equation, which is less known than the quantum Yang-Baxter equation, takes its origin in \cite{GatevaIvanova1998} and, since then, has been the interest of various algebraic studies \cite{GatevaIvanova2008,Rump2007}.

When the sets $S_i$ are finite, any solution $R_{ij}$ of the set-theoretic Yang-Baxter equation provides a solutions $\ti{R}$ of the quantum Yang-Baxter equation on $V_i=\setK^{S_i}$ defined, for any function $f\in \setK^{S_i}\otimes\setK^{S_j}\simeq \setK^{S_i\times S_j}$ and any $(y,x)\in S_j \times S_i$, by \begin{align*}
(\ti{R}_{ij}f)(y,x) &= \sum_{(x',y')\in R_{ij}^{-1}(y,x)}  f( x',y')
\end{align*}
Whenever $R_{ij}$ is a bijection, the matrix of $\ti{R}$ in the canonical basis is a permutation matrix and $\ti{R}_{ij}f = f\circ \ti{R}_{ij}^{-1}$.

\subsection{Set-theoretical Yang-Baxter for numeration basis}
We now formulate the second main theorem of the present paper on the functions $\psi_{p,q}$ . 

\begin{theo}[Yang-Baxter for \mrbs]\label{theo:YBnumbers}
For any integers $p_1$, $p_2$ and $p_3$ strictly larger than $1$, the three maps $(\psi_{p_1,p_2},\psi_{p_1,p_3},\psi_{p_2,p_3})$ satisfy the set-theoretic Yang Baxter equation
\begin{equation}\label{eq:YBpsi}
\iota_1(\psi_{p_2,p_3})\circ \iota_2(\psi_{p_1,p_3}) \circ \iota_1(\psi_{p_1,p_2}) =\iota_2(\psi_{p_1,p_2})\circ \iota_1(\psi_{p_1,p_3}) \circ 	\iota_2(\psi_{p_2,p_3}).
\end{equation}
on the three sets $S_i=\{0,1,\ldots,p_i-1\}$, $1\leq i\leq 3$.
\end{theo}
\begin{proof}
	This is again a consequence of lemma~\ref{lemm:fundanumbers}. We consider the \mrb given by $\gr{b}=(p_1,p_2,p_3)$ which encodes the decomposition of any number $0\leq n<p_1p_2p_3$ and is endowed with an action of $\Sym{3}$. We consider the permutation $c(i)=4-i$, which relates the two decompositions in $\gr{b}$ and $c(\gr{b})$:
	\begin{align*}
		n &= a_0 +a_1 p_1 + a_2 p_1p_2 & (0\leq a_i<p_{i+1})
		\\
		n &= a'_0 +a'_1 p_3 + a'_2 p_3p_2 & (0\leq a'_i<p_{3-i})
	\end{align*}
	The two braid decompositions $c= \tau_{12}\circ\tau_{23}\circ \tau_{12}$ and $c=\tau_{23}\circ \tau_{12}\circ \tau_{23}$ of $c$ provide the following two sequences of \mrbs
	\begin{align*}
		\gr{b} =(p_1,p_2,p_3) &\mapsto \tau_{12}(\gr{b})=(p_2,p_1,p_3) \mapsto \tau_{23}\circ \tau_{12}(\gr{b})=(p_2,p_3,p_1) \mapsto 
		c(\gr{b})=(p_3,p_2,p_1) 
		\\
		\gr{b} =(p_1,p_2,p_3) &\mapsto \tau_{23}(\gr{b})=(p_1,p_3,p_2) \mapsto \tau_{12}\circ \tau_{23}(\gr{b})=(p_3,p_1,p_2) \mapsto 
		c(\gr{b})=(p_3,p_2,p_1) 
	\end{align*}
	
	At the level of decompositions on the \mrbs, lemma~\ref{lemm:fundanumbers} then provides the following two sequences of transformations
	\begin{align*}
		S_1\times S_2\times S_3 &
		\xrightarrow{\iota_1(\psi_{p_1,p_2})} S_2\times S_1\times S_3 
		\xrightarrow{\iota_2(\psi_{p_1,p_3})} S_2\times S_3\times S_1 
		\xrightarrow{\iota_1(\psi_{p_2,p_3})} S_3\times S_2\times S_1 
		\\
		S_1\times S_2\times S_3 &
		\xrightarrow{\iota_2(\psi_{p_2,p_3})} S_1\times S_3\times S_2 
		\xrightarrow{\iota_1(\psi_{p_1,p_3})} S_3\times S_1\times S_2
		\xrightarrow{\iota_2(\psi_{p_1,p_2})} S_3\times S_2\times S_1 
	\end{align*}
	which decomposes the same map associated to $c$. We thus have the expected set-theoretic Yang-Baxter equation \eqref{eq:YBpsi}.
\end{proof}
There is a graphical interpretation of this Yang-Baxter equation, that we can embed in a larger set of \mrbs. Let $\gr{b}=(p_1,\ldots,p_N)$ be a finite \mrb. We now generalize the set introduced~\eqref{eq:interpmrb} to an arbitrary sequence of numbers $p_i$ and integers $l_i\geq 1$ through:
\[
\ca{B}_{p_1,\ldots,p_N}(l_1,\ldots,l_N) = \left\{\gr{b}\in \{p_1,\ldots,p_N\}^{L} ; \forall j\in\{1,\ldots,N\}, |\{1\leq i\leq L ; b_i=p_j\}|=l_j \right\}
\]
where $L=\sum_{i=1}^N l_i$. To any oriented path $\gamma_{\gr{b}}$ of length $L$ in $\prod_{i=1}^N \{0,\ldots,l_i\}$ with only length-one increasing increments, we associated a \mrb in $\ca{B}_{p_1,\ldots,p_N}(l_1,\ldots,l_N)$ by defining $b_i=p_j$ if the $i$-th increment in in direction $j$. This correspondence is bijective whenever the $p_i$ are distinct integers. The symmetric group $\Sym{L}$ acts on the paths by permuting the order of the increments and this action is compatible with the one on the \mrbs. A path is invariant under $\sigma\in\Sym{L}$ whenever each cycles of $\sigma$ acts on increments along the same dimension. From this point of view, any path in $\prod_{i=1}^N \{0,\ldots,l_i\}$ can be viewed as the projection of a path $\ca{B}_{\gr{b'}}(1,\ldots,1)$ where $\gr{b'}$ is obtained by repeating $l_i$ times each element $p_i$. For example, a square $\{0,1\}^2$ with $(p_1,p_2)=(p,p)$ is flattened to $\{0,1,2\}$. This same projective structure can be extended at the level of decompositions along \mrbs by observing that $\psi_{p,p}(x,y)=(y,x)$.

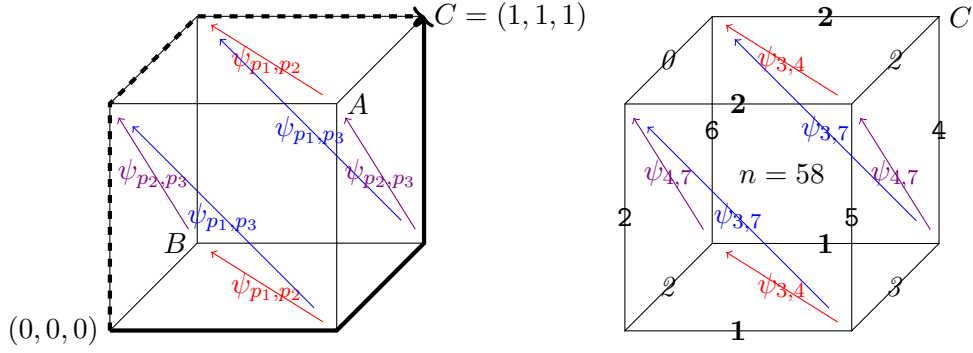
\begin{figure}
	\begin{center}
	\begin{tikzpicture}[scale=3,baseline={(current bounding box.center)}]
		\draw (0,0,0) -- (0,0,1)--(0,1,1) -- (1,1,1);
		\draw (0,0,0) -- (0,1,0)--(1,1,0) -- (1,1,1);
		\draw (0,0,0) -- (1,0,0)--(1,0,1) -- (1,1,1);
		\draw (0,0,1)--(1,0,1)  (0,1,0)--(0,1,1) (1,0,0)--(1,1,0);
		\node at (0,0,1) [left] {$(0,0,0)$};
		\node at (1,1,0) [right] {$C=(1,1,1)$};
		\node at (1,1,1) [right] {$A$};
		\node at (0,0,0) [left] {$B$};
		\draw[->,red] (0.9,0.,0.9) -- node [midway] {$\psi_{p_1,p_2}$} (0.1,0.,0.1);
		\draw[->,red] (0.9,1,0.9) -- node [midway] {$\psi_{p_1,p_2}$} (0.1,1,0.1);
		\draw[->,blue] (0.9,0.1,0)-- node [midway] {$\psi_{p_1,p_3}$} (0.1,0.9,0);
		\draw[->,blue] (0.9,0.1,1)-- node [midway] {$\psi_{p_1,p_3}$} (0.1,0.9,1);
		\draw[->,violet] (1,0.1,0.1)-- node [midway] {$\psi_{p_2,p_3}$} (1,.9,.9);
		\draw[->,violet] (0,0.1,0.1)-- node [midway] {$\psi_{p_2,p_3}$} (0,.9,.9);
		\draw[->,ultra thick] (0,0,1)--(1,0,1)--(1,0,0)--(1,1,0);
		\draw[->,ultra thick,dashed] (0,0,1)--(0,1,1)--(0,1,0)--(1,1,0);
	\end{tikzpicture}
	\begin{tikzpicture}[scale=3,baseline={(current bounding box.center)}]
	\draw (0,0,0) -- (0,0,1)--(0,1,1) -- (1,1,1);
	\draw (0,0,0) -- (0,1,0)--(1,1,0) -- (1,1,1);
	\draw (0,0,0) -- (1,0,0)--(1,0,1) -- (1,1,1);
	\draw (0,0,1)--(1,0,1)  (0,1,0)--(0,1,1) (1,0,0)--(1,1,0);
	\draw[->,red] (0.9,0.,0.9) -- node [midway] {$\psi_{3,4}$} (0.1,0.,0.1);
	\draw[->,red] (0.9,1,0.9) -- node [midway] {$\psi_{3,4}$} (0.1,1,0.1);
	\draw[->,blue] (0.9,0.1,0)-- node [midway] {$\psi_{3,7}$} (0.1,0.9,0);
	\draw[->,blue] (0.9,0.1,1)-- node [midway] {$\psi_{3,7}$} (0.1,0.9,1);
	\draw[->,violet] (1,0.1,0.1)-- node [midway] {$\psi_{4,7}$} (1,.9,.9);
	\draw[->,violet] (0,0.1,0.1)-- node [midway] {$\psi_{4,7}$} (0,.9,.9);
	\node at (0.5,0.5,0.5) {$n=58$};
	\node at (1,1,0) [right] {$C$};
	\node at (0.5,0,1) {$\mathbf{1}$};
	\node at (0.5,0,0) {$\mathbf{1}$};
	\node at (0.5,1,1) {$\mathbf{2}$};
	\node at (0.5,1,0) {$\mathbf{2}$};
	\node at (1,0,0.5) {$\mathit{3}$};
	\node at (0,0,0.5) {$\mathit{2}$};
	\node at (1,1,0.5) {$\mathit{2}$};
	\node at (0,1,0.5) {$\mathit{0}$};
	\node at (0,0.5,1) {$\mathtt{2}$};
	\node at (0,0.5,0) {$\mathtt{6}$};
	\node at (1,0.5,1) {$\mathtt{5}$};
	\node at (1,0.5,0) {$\mathtt{4}$};
	\end{tikzpicture}
	\end{center}
	\caption{\label{fig:cubeforYB}Cube associated to $\ca{B}_{p_1,p_2,p_3}(1,1,1)$: the maps $\psi_{p_i,p_j}$ are associated to faces in the dimensions $i$ and $j$. The two extremal paths $\gamma_{p_1,p_2,p_3}$ and $\gamma_{p_3,p_2,p_1}$ are represented respectively by the solid and the dashed lines. On the right, decomposition of $n=59<84=3\cdot 4\cdot 7$ in the \mrbs associated to $(p_1,p_2,p_3)=(\mathbf{3},\mathit{4},\mathtt{7})$ (same font for digits on the edges).}
\end{figure}

The Yang-Baxter equation~\eqref{eq:YBpsi} can be represented on the cube $\{0,1\}^3$ as in figure~\ref{fig:cubeforYB}. The two extremal paths $\gamma_{p_1,p_2,p_3}$ and $\gamma_{p_3,p_2,p_1}$ can be related by face transposition either by using the point $A$ or the point $B$: each sequence of face transposition corresponds to a member of \eqref{eq:YBpsi}. The two members of \eqref{eq:YBpsi} corresponds to the same cube seen from point $A=(1,1,0)$ and point $B=(0,0,1)$ as follows:
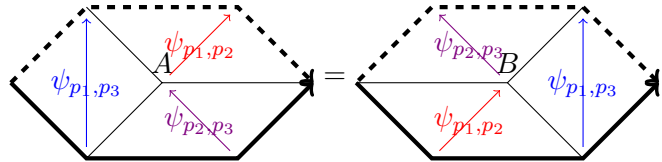
\[
\begin{tikzpicture}[baseline={(current bounding box.center)}]
	\draw[->,ultra thick] (0,0) --(1,-1) --(3,-1) --(4,0);
	\draw[->,ultra thick, dashed] (0,0) -- (1,1)--(3,1)--(4,0);
	\draw (1,-1)--(2,0) (1,1)--(2,0) (2,0)--(4,0);
	\node at (2,0) [above] {$A$};
	\draw[->,violet] (2.9,-0.9) -- node [midway] {$\psi_{p_2,p_3}$} (2.1,-0.1); 
	\draw[->,blue] (1,-0.85) -- node [midway] {$\psi_{p_1,p_3}$} (1,0.85);
	\draw[->,red] (2.1,0.1) -- node [midway] {$\psi_{p_1,p_2}$} (2.9,0.9);
\end{tikzpicture}
=
\begin{tikzpicture}[xscale=-1,baseline={(current bounding box.center)}]
	\draw[<-,ultra thick] (0,0) --(1,-1) --(3,-1) --(4,0);
	\draw[<-,ultra thick, dashed] (0,0) -- (1,1)--(3,1)--(4,0);
	\draw (1,-1)--(2,0) (1,1)--(2,0) (2,0)--(4,0);
	\node at (2,0) [above] {$B$};
	\draw[->,red] (2.9,-0.9) -- node [midway] {$\psi_{p_1,p_2}$} (2.1,-0.1); 
	\draw[->,blue] (1,-0.85) -- node [midway] {$\psi_{p_1,p_3}$} (1,0.85);
	\draw[->,violet] (2.1,0.1) -- node [midway] {$\psi_{p_2,p_3}$} (2.9,0.9);
\end{tikzpicture}
\]
Moreover, each path on the cube or in one of the hexagons above corresponds to the decomposition of the \emph{same} integer $0\leq n <p_1p_2p_3$ in the suitable \mrb as illustrated in figure~\ref{fig:cubeforYB}.

\subsection{Set-theoretical Yang-Baxter equation for polynomials} 

In the context of polynomials, the same theorem as theorem~\ref{theo:YBnumbers} can be formulated with the same mrthod of proof and it is interesting to notice that the set-theoretic Yang-Baxter equation below corresponds to an already well-known braid-like equation.

\begin{theo}[Yang-Baxter for polynomial \mrbs]\label{theo:YBpolynomial}
	For any complex numbers $a_1$, $a_2$ and $a_3$, the three maps $(\phi_{a_1,a_2},\phi_{a_1,a_3},\phi_{a_2,a_3})$ introduced in \eqref{eq:def:phi} satisfy the set-theoretic Yang Baxter equation
	\begin{equation}\label{eq:YBphi}
		\iota_1(\phi_{a_2,a_3})\circ \iota_2(\phi_{a_1,a_3}) \circ \iota_1(\phi_{a_1,a_2}) =\iota_2(\phi_{a_1,a_2})\circ \iota_1(\phi_{a_1,a_3}) \circ 	\iota_2(\phi_{a_2,a_3}).
	\end{equation}
	on the sets $S_1=S_2=S_3=S=\setK$.
\end{theo}

Interpreted on $S^3=\setK^3$, \eqref{eq:YBphi} corresponds to the well-known equation
\begin{equation}
	\begin{split}
	&\begin{pmatrix}
		1 & 0 & 0 
		\\
		0 & 1 & a_2-a_3
		\\ 
		0 & 0 & 1
	\end{pmatrix}
	\begin{pmatrix}
	1 & 0 & a_3-a_1
	\\
	0 & 1 & 0
	\\ 
	0 & 0 & 1
	\end{pmatrix}
	\begin{pmatrix}
	1 & a_2-a_1 & 0 
	\\
	0 & 1 & 0
	\\ 
	0 & 0 & 1
	\end{pmatrix}
	\\
	&=
	\begin{pmatrix}
	1 & a_2-a_1 & 0 
	\\
	0 & 1 & 0
	\\ 
	0 & 0 & 1
\end{pmatrix}
	\begin{pmatrix}
	1 & 0 & a_3-a_1
	\\
	0 & 1 & 0
	\\ 
	0 & 0 & 1
\end{pmatrix}
	\begin{pmatrix}
	1 & 0 & 0 
	\\
	0 & 1 & a_2-a_3
	\\ 
	0 & 0 & 1
\end{pmatrix}
\end{split}
\end{equation}
inherited from $SL_3(\setK)$. We may now consider bases in $\ca{B}_{a_1,\ldots,a_N}(l_1,\ldots,l_N)$. The coefficients of a polynomial in any of \mrbs in this set now belongs in $\setK^N$: a change of \mrbs is then encoded by a linear transformation of the coefficients in $SL_N(\setK)$.

\section{An example of application: the Furstenberg conjecture}
\label{sec:examples}

The reader may wonder about the practical interest of \mrbs. We provide in the present section various computations around the Furstenberg conjecture for which \mrbs provide a natural representation. In particular, we see how it is related to a construction of shifts by Rudolph and how the set-theoretic Yang-Baxter equation can be used concretely to recover the invariant finite-support atomic measures as well as the invariant Lebesgue measure.

\subsection{Furstenberg's conjecture and Rudolph's approach}
\paragraph*{\Mrbs for real numbers.}

We now provide an alternative interpretation of decorations on rectangles. Let $\gr{b}$ be a mixed radix basis. For any real number $x\in[0,1)$, there exists an infinite sequence $(a_i)_{i\geq 0}$ such that, for all $i\geq 0$, $0\leq a_i<b_{i+1}$ and 
\[
x = \sum_{i=0}^{\infty} \frac{a_i}{b_{1}\ldots b_{i+1}}=\sum_{i=0}^{\infty} \frac{a_i}{\pi_{i+1}^\gr{b}}
\]
The sequence can be constructed by considering the following recursion initiated at $x_0=x$:
\[
\begin{cases}
	a_i &= \lfloor b_{i+1} x_i \rfloor
	\\
	x_{i+1} &= b_{i+1} x_i-\lfloor b_{i+1} x_i \rfloor = \{b_{i+1} x_i\}
\end{cases}
\]
for all $i\geq 0$. This correspondence is bijective whenever one forbids sequences for which we have $a_i=b_{i+1}-1$ for all sufficiently large $i$. As for integers and polynomials, decompositions along \mrbs are related by local moves and the interesting fact is that these moves are the same as the ones presented in theorem~\ref{lemm:fundanumbers}. Indeed, let $k\geq 1$. The same number $x$ can be decomposed in $\gr{b}$ and $\tau_{k,k+1}(\gr{b})$ with:
\begin{align*}
	 x & = \sum_{i=0}^{k-2} \frac{a_i}{\pi_{i+1}^\gr{b}} + \frac{a_{k-1}}{\pi_{k-1}^{\gr{b}} b_k } + \frac{a_{k}}{\pi_{k}^{\gr{b}}b_{k}b_{k+1}  } +\sum_{i>k} \frac{a_i}{\pi_{i+1}^\gr{b}}
	 \\
	  & = \sum_{i=0}^{k-2} \frac{a_i}{\pi_{i+1}^\gr{b}} + \frac{a'_{k-1}}{\pi_{k-1}^{\gr{b}}b_{k+1}  } + \frac{a'_{k}}{\pi_{k}^{\gr{b}}b_{k+1}b_k  } +\sum_{i>k} \frac{a_i}{\pi_{i+1}^\gr{b}}
\end{align*} 
with the constraint that $a_{k-1} b_{k+1} + a_k = a'_{k-1} b_k +a'_{k}$. This corresponds precisely to
\begin{equation}
	(a_{k},a_{k-1}) = \psi_{b_k,b_{k+1}}(a'_k,a'_{k-1})
\end{equation}
The same machinery of decorations can be performed as for integers but orientations have to be reversed.

We consider $\setZ_- = \{0,-1,-2,\ldots\}$ and the unoriented graph $\setZ_-^2$ with a set $E_-$ of edges given by $\{(k-1,l),(k,l)\}\sim (k-1/2,l)$ and $\{(k,l-1),(k,l)\}\sim (k,l-1/2)$ for $(k,l)\in\setZ-^2$. A $(p,q)$-decoration is a function $\beta: E_-\to \setN\}$ such that, for all $(k,l)\in\setZ_-^2$,
\[
0\leq \beta( k-1/2,l ) <p 	\qquad \text{and}\qquad 0\leq \beta( k,l-1/2 ) <q
\]
A decoration is $\psi_{p,q}$-compatible if and only if, for all $(k,l)\in\setZ_-^2$,
\[
(\beta( k-1,l-1/2\}), \beta( k-1/2,l ) 
=
\psi_{p,q}\left( \beta( k-1/2,l-1 ),  \beta( k,l-1/2)    \right)
\]
We note $\gr{D}^{\South\West}(\psi_{p,q})$ the set of of $\psi_{p,q}$-compatible decorations on $\setZ_-^2$.

An $\South\West$-oriented infinite path is a path $\gamma : \setN \to \setZ_-^2$ such that $\gamma(0)=(0,0)$ and, for all $n\in\setN$,
\[
\gamma(n+1)-\gamma(n)  \in \{ (-1,0), (0,-1) \}
\]
There is a bijection between $\South\West$-oriented infinite paths and \mrbs $\gr{b}\in \{p,q\}^{\setN^*}$ with $\gamma_{\gr{b}}(n)-\gamma_{\gr{b}}(n-1)=(-1,0)$ if and only if $b_n=p$.

A $\psi_{p,q}$-compatible decoration $\beta$ is \emph{forbidden} if there exists $\gr{b}$ such that, for sufficiently large $i$, $\beta(\{\gamma_{\gr{b}}(i),\gamma_{\gr{b}}(i-1)\})=b_{i+1}-1$. We note $\gr{D}^{\South\West}_*(\psi_{p,q})$ the set of decorations $\beta\in \gr{D}^{\South\West}(\psi_{p,q})$ that are not forbidden.

\begin{prop}\label{prop:decorationreal}
There is a bijection $B: [0,1) \to \gr{D}^{\South\West}_*(\psi_{p,q})$, $x\mapsto B_x$ such that, for any $x\in[0,1)$ and any \mrb $\gr{b}\in\{p,q\}^{\setN^*}$,
\[
n  = \sum_{i\geq 0} \frac{B_x( \{\gamma_\gr{b}(i),\gamma_\gr{b}(i+1)\} ) }{\pi_{i+1}^{\gr{b}}}
\]
\end{prop}
\begin{proof}
	The proof follows the same steps as for integers and polynomials and is not presented here.
\end{proof}

An application of the representation $B_x$ of a real number $x$ in \emph{all} \mrbs in $\{p,q\}^{\setN^*}$ simultaneously can be found in the study of the Furstenberg conjecture stated in \cite{Furstenberg_1967} about dynamical systems.

Let $(E,\ca{E})$ be a measurable set $(E,\ca{E})$ and $T: E\to E$ a measurable map. A measure $\mu$ on $(E,\ca{E})$ is \emph{invariant} under $T$ if and only if, for all $A\in\ca{E}$, $\mu( T^{-1}(A)) =\mu(A)$.

\begin{conj}[Furstenberg]\label{conj:Furstenberg}
Let $p$ and $q$ be integers strictly larger than $1$ with $\log p/\log q\notin \setQ$. The only probability measures on $\setR/\setZ$ (with its Borel $\sigma$-algebra) \emph{invariant} under the two multiplication maps $T_p (x) = px \mod 1$ and $T_q(x)=qx\mod 1$ are either atomic measures with finite support or the Lebesgue measure. 
\end{conj}

It is easy to see that the Lebesgue measure is indeed invariant and that the invariant atomic measures with finite support are given by the uniform measures on the orbits of a rational number with a denominator prime to $p$ and $q$. The maps $T_p$ and $T_q$ can be lifted at the level of allowed decorations easily and corresponds then to shifts.

For any decoration $\beta\in \gr{D}^{\South\West}_*(\psi_{p,q})$, we define the following two horizontal and vertical shifts by
\begin{align*}
	\sigma_1(\beta)( \{(k,l),(k',l')\} ) &= \beta( \{(k-1,l),(k'-1,l')\} 
	\\
	\sigma_2(\beta)( \{(k,l),(k',l')\} ) &= \beta( \{(k,l-1),(k',l'-1)\}
\end{align*}
for any edge $\{(k,l),(k',l')\}$ of $\setZ_-^2$.

\begin{prop}\label{prop:fromTtoshifts}
	The two maps $T_p$ and $T_q$ on $\setR/\setZ$ satisfy, for all $x\in[0,1)$, 
	\begin{align}
	B(T_p(x)) &= \sigma_1(B(x)) 
	&
	B(T_q(x)) &= \sigma_2(B(x)) 
	\end{align}
\end{prop}
\begin{proof}
It relies that multiplication by $p$ (resp. $q$) modulo 1 have a simple interpretation in \mrbs. For any $\gr{b}\in\{p,q\}^{\setN^*}$ with $\gr{b}=(p,\gr{b'})$ (it starts with $p$), one has
$T_p(0,x_1x_2\ldots_{\gr{b}}) = 0,x_2\ldots_{\gr{b}'}$. 
\end{proof}

Using the bijection $B$, proposition~\ref{prop:fromTtoshifts} provides an equivalence between the invariant probability measures of conjecture~\ref{conj:Furstenberg} with probability measures on $\gr{D}^{\South\West}_*(\psi_{p,q})$  (endowed with its cylindrical $\sigma$-algebra) invariant under both shifts $\sigma_1$ and $\sigma_2$. The shift invariance of such a measure $\mu$ on $\gr{D}^{\South\West}_*(\psi_{p,q})$ means that, for any sub-graph $Q_{k,l}=\{\ldots,-k-2,-k-1,-k\}\times  \{\ldots,-l-2,-l-1,-l\}$, the shift $(x,y)\mapsto (x+k,y+l)$ from $Q_{k,l}$ to $\setZ_-^2$ of the restriction to $Q_{k,l}$ of a random decoration $\beta$ under the probability measure $\mu$ has the same probability measure $\mu$ as the initial decoration $\beta$.

\paragraph*{Relation with Rudolph's construction.} Furstenberg's conjecture is still open but major steps have been realized by Rudolph \cite{Rudolph_1990}, Host \cite{Host_1995} and Parry \cite{Parry_1996}. In particular, the present construction coincide with Rudolph's two-dimensional arrays $Y\subset \Sigma^{\setN^2}$ in \cite{Rudolph_1990} through the following explicit correspondence. A decoration $\beta\in \gr{D}^{\South\West}_*(\psi_{p,q})$ is a map defined on the edges of $\setZ_-^{2}$. We consider each face of $\setZ_-^{2}$, which corresponds to an elementary square of length $1$, and label it by its centre in $(\setZ_--1/2)^2$. Each face $(i,j)$ is surrounded by four edges such that, by construction,
\[
\beta(i,j-1/2) + p\beta(i+1/2,j) = \beta(i-1/2,j)+q\beta(i,j+1/2) \in \{0,1,\ldots,pq-1\}
\]
We note $\Upsilon(i,j)$ this integer. After relabelling $(i,j)\mapsto (-i-1/2,-j-1/2)\in\setN^2$, $\Upsilon$ is precisely the array in $Y$ introduced by Rudolph and indeed we observe that the diagonal part of $\Upsilon$ provides the decomposition of $x\in [0,1)$ in the constant radix $pq$ basis, i.e. $x=\sum_{n\geq 1} \Upsilon(-n+1/2,-n+1/2) (pq)^{-n}$ and the non-diagonal parts the decompositions of $T_p^{n}T_q^{m}(x)$.

The interest of our representation is double. The first one is that it localizes information on the edges and not on the faces: this allows for the use of separating oriented paths $\gamma_\gr{b}$ as in lemma~\ref{lemm:filling}, which corresponds to a spatial Markov property suitable for a study as in \cite{Simon2023} (see below) for the construction of infinite-volume Gibbs measure related to Furstenberg's conjecture. The spatial Markov property corresponds to the following fact: a path $\gamma_{\gr{b}}$ separates $\setZ_-^2$ into two components and the decorations in the two components depend only on the decoration along the common boundary $\gamma_{\gr{b}}$ of the two components.

The second interest is that the present framework of \mrbs introduces the tool of the Yang-Baxter equation relate to integrable models. Integrable models are known to exhibit both a strong rigidity and exact computations and we expect that this rigidity is related to the one described in Furstenberg's conjecture and that the exact computations corresponds to the fact that the conjectured measures are indeed explicitly described.

We illustrate these two concepts by rewriting the invariant measures listed in conjecture~\ref{conj:Furstenberg} in terms of the set-theoretic Yang-Baxter equation described in section~\ref{sec:YB}.

\subsection{Relation of some invariant ergodic measures with Yang-Baxter equation}\label{sec:layer}
\paragraph*{The case of atomic measures with finite support.}
Let $k/r$ be a rational number in $[0,1)$ with $k$ and $r$ prime to each other and such that $r$ is prime with both $p$ and $q$. The orbit of $k/r$ under the maps $T_p$ and $T_q$ is finite and we write it as $\Orb_{p,q}(k/r)=\{k_1/r,k_2/r,\ldots,k_d/r\}$. For any $1\leq i\leq d$, we thus have $T_p(k_i)=k_{\tau_p(i)}$ and $T_q(k_i)=k_{\tau_{q}(i)}$ for some integers $\tau_p(i)$ and $\tau_q(i)$. We endow this orbit with the uniform measure. It is then an invariant and ergodic measure for $T_p$ and $T_q$.

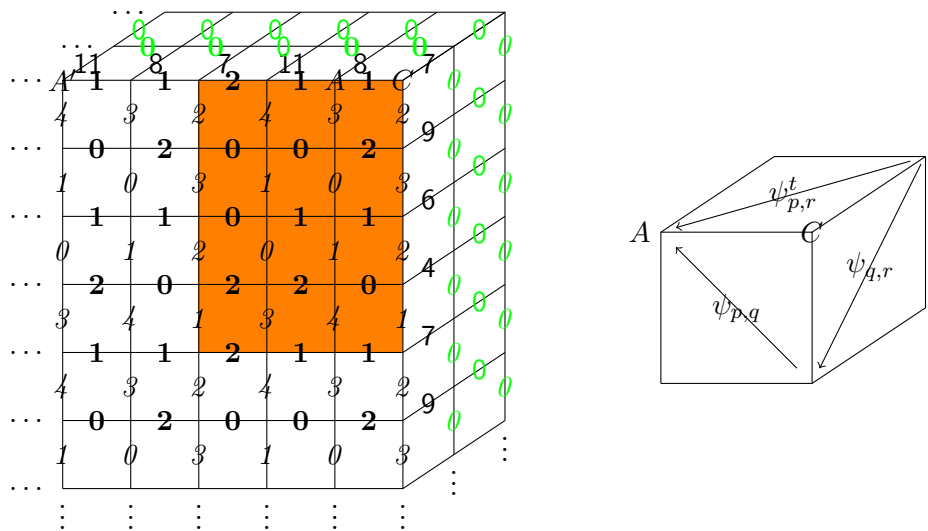
An interesting construction consists in the introduction of the enriched \mrbs $\ca{B}_{p,q,r}=\{p,q,r\}^{\setN^*}$ with a three-dimensional corner representation on $\setZ_-^3$ of any real $x\in[0,1)$ following an easy generalization of proposition~\ref{prop:decorationreal}. In this case, an element $k_i/r \in \Orb_{p,q}(k/r)$ has the representation $0,k_i00\ldots\gr{b}$ in any basis $\gr{b}$ with $b_1=r$, i.e. only the first digit is non zero. An example is presented in figure~\ref{fig:example3513}. As illustrated in this figure, Yang-Baxter transformations cannot be applied directly from the origin since the arrows are not in the correct order. However, truncated the infinite corner to a box shape and considering the box seen from the point $A'$ as in figure~\ref{fig:example3513} allows the use of the Yang-Baxter equation.

\begin{figure}
	\begin{center}
		\begin{tikzpicture}[scale=0.9,baseline={(current bounding box.center)}]
		\fill[orange,opacity=0.2] (0,0) rectangle (-3,-4);
		\foreach \y in {0,-1,...,-6} 
		{
			\draw (-5,\y) -- (0,\y) -- (1.5,\y+1);	
			\node at (-5.5,\y) {$\dots$};
		}
		\foreach \x in {-5,-4,...,0} 
		{
			\draw (\x,-6) --(\x,0) -- (\x+1.5,1);
			\node at (\x,-6.3) {$\vdots$};	
		}
		\draw (0.75,-5.5)--(0.75,0.5) --(-4.25,0.5);
		\node at (0.75,-5.8) {$\vdots$};
		\node at (-4.75,0.5) {$\dots$};
		\draw (1.5,-5.)--(1.5,1) --(-3.5,1);
		\node at (1.5,-5.3) {$\vdots$};
		\node at (-4,1) {$\dots$};
		\foreach \y in {0,-4} {
			\node at (-0.5,\y) {$\mathbf{1}$};
			\node at (-1.5,\y) {$\mathbf{1}$};
			\node at (-2.5,\y) {$\mathbf{2}$};
			\node at (-3.5,\y) {$\mathbf{1}$};
			\node at (-4.5,\y) {$\mathbf{1}$};
			\node at (0.375,\y+0.25) {$\mathsf{7}$};
			\node at (0.375+0.75,\y+0.75)[green] {$\mathsf{0}$};
			\node at (0,\y-0.5) {$\mathit{2}$};
			\node at (-1,\y-0.5) {$\mathit{3}$};
			\node at (-2,\y-0.5) {$\mathit{4}$};
			\node at (-3,\y-0.5) {$\mathit{2}$};
			\node at (-4,\y-0.5) {$\mathit{3}$};
			\node at (-5,\y-0.5) {$\mathit{4}$};
			\node at (0.75,\y-0.5+0.5) [green]{$\mathit{0}$};
			\node at (1.5,\y-0.5+1)[green] {$\mathit{0}$};
		}
		\foreach \y in {-1,-5} {
			\node at (-0.5,\y) {$\mathbf{2}$};
			\node at (-1.5,\y) {$\mathbf{0}$};
			\node at (-2.5,\y) {$\mathbf{0}$};
			\node at (-3.5,\y) {$\mathbf{2}$};
			\node at (-4.5,\y) {$\mathbf{0}$};
			\node at (0.375,\y+0.25) {$\mathsf{9}$};
			\node at (0.375+0.75,\y+0.75)[green] {$\mathsf{0}$};
			\node at (0,\y-0.5) {$\mathit{3}$};
			\node at (-1,\y-0.5) {$\mathit{0}$};
			\node at (-2,\y-0.5) {$\mathit{1}$};
			\node at (-3,\y-0.5) {$\mathit{3}$};
			\node at (-4,\y-0.5) {$\mathit{0}$};
			\node at (-5,\y-0.5) {$\mathit{1}$};
			\node at (0.75,\y-0.5+0.5)[green] {$\mathit{0}$};
			\node at (1.5,\y-0.5+1) [green]{$\mathit{0}$};
		}
		\foreach \y in {-2} {
			\node at (-0.5,\y) {$\mathbf{1}$};
			\node at (-1.5,\y) {$\mathbf{1}$};
			\node at (-2.5,\y) {$\mathbf{0}$};
			\node at (-3.5,\y) {$\mathbf{1}$};
			\node at (-4.5,\y) {$\mathbf{1}$};
			\node at (0.375,\y+0.25) {$\mathsf{6}$};
			\node at (0.375+0.75,\y+0.75) [green]{$\mathsf{0}$};
			\node at (0,\y-0.5) {$\mathit{2}$};
			\node at (-1,\y-0.5) {$\mathit{1}$};
			\node at (-2,\y-0.5) {$\mathit{0}$};
			\node at (-3,\y-0.5) {$\mathit{2}$};
			\node at (-4,\y-0.5) {$\mathit{1}$};
			\node at (-5,\y-0.5) {$\mathit{0}$};
			\node at (0.75,\y-0.5+0.5)[green] {$\mathit{0}$};
			\node at (1.5,\y-0.5+1) [green]{$\mathit{0}$};
		}
		\foreach \y in {-3} {
			\node at (-0.5,\y) {$\mathbf{0}$};
			\node at (-1.5,\y) {$\mathbf{2}$};
			\node at (-2.5,\y) {$\mathbf{2}$};
			\node at (-3.5,\y) {$\mathbf{0}$};
			\node at (-4.5,\y) {$\mathbf{2}$};
			\node at (0.375,\y+0.25) {$\mathsf{4}$};
			\node at (0.375+0.75,\y+0.75) [green]{$\mathsf{0}$};
			\node at (0,\y-0.5) {$\mathit{1}$};
			\node at (-1,\y-0.5) {$\mathit{4}$};
			\node at (-2,\y-0.5) {$\mathit{3}$};
			\node at (-3,\y-0.5) {$\mathit{1}$};
			\node at (-4,\y-0.5) {$\mathit{4}$};
			\node at (-5,\y-0.5) {$\mathit{3}$};
			\node at (0.75,\y-0.5+0.5) [green] {$\mathit{0}$};
			\node at (1.5,\y-0.5+1)[green] {$\mathit{0}$};
		}
		\node at (-1+0.375,0.25) {$\mathsf{8}$};
		\node at (-2+0.375,0.25) {$\mathsf{11}$};
		\node at (-3+0.375,0.25) {$\mathsf{7}$};
		\node at (-4+0.375,0.25) {$\mathsf{8}$};
		\node at (-5+0.375,0.25) {$\mathsf{11}$};
		\node at (-1+0.375+0.75,0.75) [green]{$\mathsf{0}$};
		\node at (-2+0.375+0.75,0.75) [green]{$\mathsf{0}$};
		\node at (-3+0.375+0.75,0.75) [green]{$\mathsf{0}$};
		\node at (-4+0.375+0.75,0.75) [green]{$\mathsf{0}$};
		\node at (-5+0.375+0.75,0.75) [green]{$\mathsf{0}$};
		\node at (-1+0.75+0.5,0.5) [green]{$\mathbf{0}$};
		\node at (-2+0.75+0.5,0.5) [green]{$\mathbf{0}$};
		\node at (-3+0.75+0.5,0.5) [green]{$\mathsf{0}$};
		\node at (-4+0.75+0.5,0.5) [green]{$\mathbf{0}$};
		\node at (-5+0.75+0.5,0.5) [green]{$\mathbf{0}$};
		\node at (0,0) {$C$};
		\node at (-5,0) {$A'$};
		\node at (-1,0) {$A$};
		\end{tikzpicture}
	\hspace{1cm}
	\begin{tikzpicture}[scale=2,baseline={(current bounding box.center)}]
		\draw (0,0)--(1,0)--(1.75,0.5) --(1.75,1.5);
		\draw (0,0)--(0,1)--(0.75,1.5)--(1.75,1.5);
		\draw (1,0)--(1,1)--(1.75,1.5) (0,1)--(1,1);
		\draw[->] (0.9,0.1) --node [midway] {$\psi_{p,q}$} (0.1,0.9);
		\draw[->] (1.65,1.47) -- node [midway] {$\psi^t_{p,r}$}(0.1,1.03);
		\draw[->] (1.72,1.45)-- node [midway] {$\psi_{q,r}$}(1.05,0.1);
		\node at (1,1) {$C$};
		\node at (0,1) [left] {$A$};
	\end{tikzpicture}
	\end{center}
	\caption{\label{fig:example3513}Decomposition of $7/13$ in the \mrb $\ca{B}_{3,5,13}=\{3,5,13\}^{\setN^*}$. One observes a periodicity of length $3$ in the horizontal dimension associated to $p=3$ (bold digits) and a periodicity of length $4$ in the vertical dimension associated to $q=5$ (italic digits). In the third dimension associated to $13$, all the digits are zero beyond the first ones in the first layer. The periodic horizontal sequence $(\mathsf{7},\mathsf{8},\mathsf{11},\mathsf{7},\ldots)$ corresponds to the orbit of $7/13$ under the map $T_3$. The vertical sequence $(\mathsf{7},\mathsf{9},\mathsf{6},\mathsf{4},\mathsf{7},\ldots)$ corresponds to the orbit of $7/13$ under the map $T_5$. The cube corner on the right is seen from the vertex $C$, in comparison with figure~\ref{fig:cubeforYB} where it was seen from the point $A$. In particular, there is a transposition of the arguments, indicated by an exponent $t$ between the left and right edge decorations with respect to the arrow of the transformation.}
\end{figure}

Generating all decorations for $B(k_i/r)$ can be done in the following way. We consider as input the unique value $k_i$ on the edge between $(0,0,0)$ and $(0,0,-1)$ (the $\mathsf{7}$ on the figure \ref{fig:example3513}) and all the zero digits on the second layer, i.e. all the digits on the edges in the plane $\setZ\times\setZ\times \{-1\}$. This corresponds to the fact that the number $k_i/r$ can be written as $0,k_i0\ldots0_{\gr{b}})$ in \emph{all} the \mrbs with $\gr{b}=(r,\gr{b'})$ and $\gr{b'}\in\{p,q\}^{\setN^*}$. Seen from the point $A$, this corresponds to the following input configuration
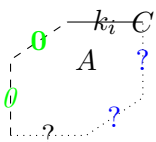
\[
\begin{tikzpicture}[scale=1]
	\draw[dashed] (0,0) -- (0,1) -- (0.75,1.5);
	\draw (0.75,1.5)-- (1.75,1.5);
	\node at (0,0.5) [green] {$\mathit{0}$};
	\node at (0.375,1.25) [green] {$\mathbf{0}$};
	\node at (1.25,1.5) {$k_i$};
	\draw[dotted] (0,0) -- (1,0) -- (1.75,0.5) -- (1.75,1.5);
	\node at (0.5,0) {$?$};
	\node at (1.375,0.25) [blue]{$?$};
	\node at (1.75,1) [blue] {$?$};
	\node at (1.75,1.5)  {$C$};
	\node at (1,1) {$A$};
\end{tikzpicture}
\]
where the edges of the second layer appears as a dashed line and all digits $0$, $k_i$ and $\beta(\cdot,\cdot)$ are variables associated to edges. We now apply one of the hand side of the set-theoretical Yang-Baxter equation \eqref{eq:quantumYB} to this input configuration and obtain (all variables on the edges of the cube):
\begin{equation}\label{eq:YBlayer1}
\begin{tikzpicture}[scale=1.7,baseline={(current bounding box.center)}]
	\draw[dashed] (0,0) -- (0,1) -- (0.75,1.5) ;
	\draw (0.75,1.5) -- (1.75,1.5);
	\node at (0,0.5) [green] {$\mathit{0}$};
	\node at (0.375,1.25) [green] {$\mathbf{0}$};
	\node at (1.25,1.5) {$k_i$};
	\draw (0,0) -- (1,0) -- (1.75,0.5) -- (1.75,1.5);
	\node at (0.5,0) {$?$};
	\node at (1.375,0.25) [blue]{$?$};
	\node at (1.75,1) [blue] {$?$};
	\draw[dashed] (0,0) -- node [midway,green] {$\mathbf{0}$} (0.75,0.5) ;
	\draw (0.75,0.5)-- node [midway] {$?$}(1.75,0.5);
	\draw[dashed] (0.75,0.5) -- node [midway,green] {$\mathit{0}$} (0.75,1.5);
\end{tikzpicture}
\to 
\begin{tikzpicture}[scale=1.7,baseline={(current bounding box.center)}]
	\draw[dashed] (0,0) -- (0,1) -- (0.75,1.5);
	\draw (0.75,1.5) -- (1.75,1.5);
	\node at (0,0.5) [green] {$\mathit{0}$};
	\node at (0.375,1.25) [green] {$\mathbf{0}$};
	\node at (1.25,1.5) {$k_i$};
	\draw (0,0) -- (1,0) -- (1.75,0.5) -- (1.75,1.5);
	\node at (0.5,0) {$?$};
	\node at (1.375,0.25) [blue]{$?$};
	\node at (1.75,1) [blue] {{\scriptsize$\beta(0,-0.5)$}};
	\draw[dashed] (0,0) -- node [midway,green] {$\mathbf{0}$} (0.75,0.5) ;
	\draw (0.75,0.5)-- node [midway] {{\footnotesize$k_{\tau_q(i)}$}}(1.75,0.5);
	\draw[dashed] (0.75,0.5) -- node [midway,green] {$\mathit{0}$} (0.75,1.5);
\end{tikzpicture}
\to 
\begin{tikzpicture}[scale=1.7,baseline={(current bounding box.center)}]
	\draw[dashed] (0,0) -- (0,1) -- (0.75,1.5) ;
	\draw (0.75,1.5) -- (1.75,1.5);
	\node at (0,0.5) [green] {$\mathit{0}$};
	\node at (0.375,1.25) [green] {$\mathbf{0}$};
	\node at (1.25,1.5) {$k_i$};
	\draw (0,0) -- (1,0) -- (1.75,0.5) -- (1.75,1.5);
	\node at (0.5,0) {{\footnotesize$k_{\tau_p(\tau_q(i))}$}};
	\node at (1.375,0.25) [blue]{{\scriptsize$\beta(-0.5,-1)$}};
	\node at (1.75,1) [blue]{{\scriptsize$\beta(0,-0.5)$}};
	\draw[dashed] (0,0) -- node [midway,green] {$\mathbf{0}$} (0.75,0.5) ;
	\draw (0.75,0.5) -- node [midway] {{\footnotesize$k_{\tau_q(i)}$}}(1.75,0.5);
	\draw[dashed] (0.75,0.5) -- node [midway,green] {$\mathit{0}$} (0.75,1.5);
\end{tikzpicture}
\end{equation}
We first observe the important fact that, on the second layer only zero values are produced by the first $\psi_{q,p}$ map, which is compatible with our hypothesis of only zero digits on the second layer. We also observe that only elements of the orbit $\Orb_{p,q}(k_i/r)$ appear in the edges between the two layer. On the first layer, two decorations are generated, which are illustrated on figure~\ref{fig:example3513} as the two digits $\mathit{2}$ and $\mathbf{2}$ on the corner cube.

Starting from the same configuration, we may now apply the second handside of the Yang-Baxter equation \eqref{eq:quantumYB} and now produce the following sequence of fillings:
\begin{equation}\label{eq:YBlayer2}
\begin{tikzpicture}[scale=1.7,baseline={(current bounding box.center)}]
	\draw[dashed] (0,0) -- (0,1) -- (0.75,1.5) ;
	\draw (0.75,1.5) -- (1.75,1.5);
	\node at (0,0.5) [green] {$\mathit{0}$};
	\node at (0.375,1.25) [green] {$\mathbf{0}$};
	\node at (1.25,1.5) {$\mathsf{k_i}$};
	\draw (0,0) -- (1,0) -- (1.75,0.5) -- (1.75,1.5);
	\node at (0.5,0) {$?$};
	\node at (1.375,0.25) [blue]{$?$};
	\node at (1.75,1) [blue]{$?$};
	\draw (0,1)-- node [midway] {{\footnotesize$k_{\tau_p(i)}$}} (1,1);
	\draw (1,1)-- node [midway,blue] {{\scriptsize$\beta(-0.5,0)$}} (1.75,1.5);
	\draw (1,0)-- node [midway,blue] {$?$} (1,1);
\end{tikzpicture}
\to 
\begin{tikzpicture}[scale=1.7,baseline={(current bounding box.center)}]
	\draw[dashed] (0,0) -- (0,1) -- (0.75,1.5) ;
	\draw (0.75,1.5) -- (1.75,1.5);
	\node at (0,0.5) [green] {$\mathit{0}$};
	\node at (0.375,1.25) [green] {$\mathbf{0}$};
	\node at (1.25,1.5) {$\mathsf{k_i}$};
	\draw (0,0) -- (1,0) -- (1.75,0.5) -- (1.75,1.5);
	\node at (0.5,0) {{\footnotesize$k_{\tau_p(\tau_q(i))}$}};
	\node at (1.375,0.25) [blue] {$?$};
	\node at (1.75,1) [blue]{$?$};
	\draw (0,1)-- node [midway] {{\footnotesize$k_{\tau_p(i)}$}} (1,1);
	\draw (1,1)-- node [midway,blue] {{\scriptsize$\beta(-0.5,0)$}} (1.75,1.5);
	\draw (1,0)-- node [midway,blue] {{\scriptsize$\beta(-1,-0.5)$}} (1,1);
\end{tikzpicture}
\to 
\begin{tikzpicture}[scale=1.7,baseline={(current bounding box.center)}]
	\draw[dashed] (0,0) -- (0,1) -- (0.75,1.5) ;
	\draw (0.75,1.5) -- (1.75,1.5);
	\node at (0,0.5) [green] {$\mathit{0}$};
	\node at (0.375,1.25) [green] {$\mathbf{0}$};
	\node at (1.25,1.5) {$\mathsf{k_i}$};
	\draw (0,0) -- (1,0) -- (1.75,0.5) -- (1.75,1.5);
	\node at (0.5,0) {{\footnotesize$k_{\tau_p(\tau_q(i))}$}};
	\node at (1.375,0.25) [blue] {{\scriptsize$\beta(-0.5,-1)$}};
	\node at (1.75,1) [blue]{{\scriptsize$\beta(0,-0.5)$}};
	\draw (0,1)-- node [midway] {{\footnotesize$k_{\tau_p(i)}$}} (1,1);
	\draw (1,1)-- node [midway,blue] {{\scriptsize$\beta(-0.5,0)$}} (1.75,1.5);
	\draw (1,0)-- node [midway,blue] {{\scriptsize$\beta(-1,-0.5)$}} (1,1);
\end{tikzpicture}
\end{equation}
The two Yang-Baxter moves ensure that the second layer is consistent since zero digits are conserved and show that on the transverse edges elements of the orbits are obtained through the maps $\tau_p$ and $\tau_q$ and then consistent decorations are produced on the first layer.

A double recursion can then be applied to all the cubes between the first and second layers with $k_i$ replaced by another element in the orbit and the whole decoration $B(k_i/r)$ is generated on the first layer. There are now two key observations:
\begin{enumerate}
	\item if one starts with a random element $k_i$ chosen uniformly in the orbit, then any $k_{\tau_p^{\circ m}\tau_q^{\circ n}(i)}$ is again uniformly distributed in the orbit.
	\item the sub-decoration $\sigma_1^{\circ m}\sigma_2^{\circ n}(B(k_i/r))$ is generated by the same recursion which uses the same second layer and the starting value $k_{\tau_p^{\circ m}\tau_q^{\circ n}(i)}$.
\end{enumerate}
We now obtain that the atomic measures described previously are invariant as stated in conjecture~\ref{conj:Furstenberg}. We also remark that we may have started with $0$ not in the first layer but in the $l$-th layer: this would then correspond to orbits of $k'/r^l$ for some number $k'=k_1+rk_2+\ldots +k_{l}r^{l-1}$ with digits put on the edges between the layers.

\paragraph*{The case of the Lebesgue measure.}
The invariance of the Lebesgue measure can be described in the same way:
\begin{itemize}
	\item uniform and independent digits are chosen along a given path in the second layer (this corresponds to an element $x'\in[0,1)$) and the digits on the other edges of the second layer are then filled following section~\ref{sec:mixedradix};
	\item a uniform digit $k$ in some arbitrary basis parameter $r$ is chosen independently on the first transverse edge.
\end{itemize}
The same filling mechanism as before then ensures the generation of the whole decoration on the first layer. 

Together, such a choice corresponds to a number $x=k/r+x'/r$ for the whole octant. This number is again uniform in $[0,1)$ as it can be seen from simple probabilistic computations and thus the digits on the first layer have the same law as the digits on the second layer. It is interesting to see it also directly from the Yang-Baxter moves \eqref{eq:YBlayer1} and \eqref{eq:YBlayer2} of filling cubes described above: for each admissible law on the four digits of an elementary square of the second layer and the digit of the first transversal edge (the $k_i$ above), one obtains a new law on the four digits of the corresponding elementary square on the first layer. 

One also deduces that the shifted sub-decorations $\sigma_1^{\circ m}\sigma_2^{\circ n}(B(x))$ have the same probability distribution as the initial distribution $B(x)$, hence the invariance described in conjecture~\ref{conj:Furstenberg}.

One also observe that this mechanism can be started with an arbitrary basis parameter $r$ and from an arbitrary depth of layers with the same hypothesis on the lowest layer.

\paragraph*{Towards the Furstenberg conjecture.} We have seen that the two-dimensional view on quadrants of a number $x\in[0,1)$ using \mrbs rewrites $T_p$ and $T_q$ in terms of shifts. The additional Yang-Baxter structure provides a multi-layer approach to invariant measures: We have seen above that \emph{both} the finite-support atomic measures and the Lebesgue measure behave nicely between different layers: 
\begin{itemize}
\item in the first case, the first layer is inherited from a trivial second (or larger finite-rank) layer with zeroes everywhere and a uniform digit on a transersal edge;
\item in the second case, the first-layer Lebesgue measure is a fixed point for the layer-to-layer transformation.
\end{itemize}
The layer-to-layer transformation provides a potential starting point to the Furstenberg conjecture. Is this layer-to-layer transformation too simple to have other layer-to-layer-invariant shifts-invariant measures or does the layer point of view provides alternative infinite-layer constructions of shifts invariant measures ? 

The question is open for us now. It may be interesting to first recast Rudolph's entropy criterion in terms of the present layer-to-layer structure. In a second time, the Yang-Baxter moves are local transformation whereas the invariant measures on $[0,1)$ are global objects: it may be interesting to see how generic "global" measures on $[0,1)$ can be built out of local building blocks related to edges inside layers and between layers. Such an idea is presented below and will be explored in a further work.

\subsection{An alternative approach through the Markov property}

Localizing the information on the edges instead of faces through \mrbs provides a spatial Markov property on two-dimensional arrays labelled on the edges and thus the construction of \cite{Simon2023} can be used. The present situation corresponds to a face weight
\begin{equation}
	\begin{split}
	\MarkovWeight{R}_{p,q} : \{0,1,\ldots,p-1\}^2\times  \{0,1,\ldots,q-1\}^2 &\to\setR_+
	\\
	(u_\South,u_\North,u_\West,u_\East) & \mapsto \indic{u_\South+pu_\East = u_\West+qu_\North}
	\end{split}
\end{equation}
and variables on the edges are the digits in suitable \mrbs. As a finite state model of statistical mechanics, one may be interested in the free energy of the model as well as the definition of  associated infinite-volume Gibbs measures on the whole plane $\setZ^2$ with suitable translational invariance. Such infinite-volume Gibbs measures are constructed in \cite{Simon2023} through Kolmogorov's extension theorem applied to boundary conditions on finite rectangles submitted to suitable boundary conditions built out of operators. 

In the present case, the definitions up to morphisms introduced in \cite{Simon2023} have a set-theoretical equivalent formulation (in the same way as definitions of \cite{Simon2023} have a set-theoretical intepretation in the parameter space of \cite{Bodiot2023}). We may consider, for example on the North side, a set $E_\North$ and a map $A_{\North}: \{0,1,\ldots,p-1\}\times E_{\North} \to E_{\North}$ such that the following vertical product is an extended map:
\[
A'_\North=\begin{tikzpicture}[guillpart,yscale=1.5,xscale=1.5]
	\fill[guillfill] (0,0) rectangle (1,2);
	\draw[guillsep] (0,2)--(0,0)--(1,0)--(1,2) (0,1)--(1,1);
	\node at (0.5,0.5) {$\psi_{p,q}$};
	\node at (0.5,1.5) {$A_\North$};
\end{tikzpicture} :  \{0,1,\ldots,p-1\} \times (\{0,1,\ldots,q-1\}\times E_{\North}) \to (\{0,1,\ldots,q-1\}\times E_{\North}) 
\]
with the element-wise definition $A'_\North(u_\South, (v,a) ) = \left(\psi^{(1)}_{p,q}(u_\South,v),A_\North(\psi^{(2)}_{p,q}(u_\South,v),a)\right)$
The element $A_\North$ is then a \emph{fixed point up to morphism} of $\psi_{p,q}$ if there exists a map $S : \{0,1,\ldots,q-1\}\times E_{\North} \to E_{\North}$ such that $S \circ A'_{\North} = A_{\North} \circ (\id\times S)$. Using the layer structure introduced in section~\ref{sec:layer}, the previous Yang-Baxter construction naturally provides such elements. It would be interesting to see if the new degrees of liberty left open in \cite{Simon2023} may provide interesting invariant measures for the $T_p$ and $T_q$ dynamics.

\section*{Acknowledgements}

We thank warmly François Ledrappier for his interest in the present work, his encouragements to write down the present construction and for mentioning Furstenberg's conjecture to us. We also thank the two referees for their careful readings and very interesting questions and suggestions.

This research was partially funded by the Agence Nationale de la Recherche (ANR), grant ANR-24-CE40-7252 NASQI3D "New algebraic structures in quantum integrability: towards 3D".

\bibliographystyle{alphaurl}
\bibliography{biblio_basesYB}

\end{document}